\documentclass[12pt,a4paper]{amsart}  
\newcommand\abs[1]{\lvert #1\rvert}

\usepackage[T1]{fontenc}
\usepackage{datetime}
\usepackage{amsmath,amsthm,amssymb,enumerate,amsfonts,graphicx}
\usepackage{color}
\usepackage{epsfig}
\usepackage[numbers,sort&compress]{natbib}
\usepackage[utf8]{inputenc}
\usepackage{algorithm,algorithmic}
\usepackage{stmaryrd}
\usepackage{paralist}
\usepackage{tikz}
\usepackage{ifthen}
\usepackage{upgreek}
\usepackage{hyperref}

\tikzset{black node/.style={draw, circle, fill = black, minimum size = 5pt, inner sep = 0pt}}
\tikzset{white node/.style={draw, circle, fill = white, minimum size = 5pt, inner sep = 0pt}}
\tikzset{normal/.style = {draw=none, fill = none}}

\pgfdeclarelayer{bg}    
\pgfsetlayers{bg,main}

\usepackage{todonotes}

\hypersetup{
  pdftitle = {Packing and covering induced subdivisions, long},
  pdfauthor = {Kwon, Raymond},
  colorlinks  
}

\usepackage{aliascnt}

\newtheorem{theorem}{Theorem}[section]

\newaliascnt{lemma}{theorem}
\newtheorem{lemma}[lemma]{Lemma}
\aliascntresetthe{lemma}

\newaliascnt{proposition}{theorem}
\newtheorem{proposition}[proposition]{Proposition}
\aliascntresetthe{proposition}

\newaliascnt{corollary}{theorem}
\newtheorem{corollary}[corollary]{Corollary}
\aliascntresetthe{corollary}

\newaliascnt{conjecture}{theorem}

\aliascntresetthe{conjecture}

\newaliascnt{claim}{theorem}
\newtheorem{claim}[claim]{Claim}
\aliascntresetthe{claim}

\newaliascnt{observation}{theorem}

\aliascntresetthe{observation}

\newcommand*{\proofclaimname}{Proof}
\newenvironment{proofclaim}[1][\proofclaimname]{\begin{proof}[#1]}{\end{proof}}

\newcommand{\arxiv}[1]{\href{http://arxiv.org/abs/#1}{\tt arXiv:#1}}

\newcommand{\N}{\mathbb{N}}
\newcommand{\R}{\mathbb{R}}

\newcommand{\ceil}[1]{\left\lceil#1\right\rceil}

\newcommand{\dist}{\operatorname{dist}}
\newcommand{\last}{\operatorname{end}}

\DeclareMathOperator{\tw}{\mathbf{tw}}
\DeclareMathOperator{\girth}{\mathbf{girth}}

\newcommand{\intv}[2]{\left \{ #1, \dots, #2 \right \}}
\newcommand{\ep}{Erd\H{o}s-P\'osa}

\begin{document}

\title{Packing and covering induced subdivisions}\thanks{This research has been supported by the European Research Council (ERC) under the European Union's Horizon 2020 research and innovation program, ERC consolidator grant DISTRUCT, agreement No 648527. O-joung Kwon was also supported by the National Research Foundation of Korea (NRF) grant funded by the Ministry of Education (No. NRF-2018R1D1A1B07050294).}

\author[O.~Kwon]{O-joung Kwon}
\author[J.-F.~Raymond]{Jean-Florent Raymond}
\address[O.~Kwon]{\newline Department of Mathematics
\newline Incheon National University
\newline Incheon, South Korea}
\email{ojoungkwon@gmail.com}

\address[J.-F.~Raymond]{\newline Logic and Semantics Research Group
\newline Technische Universität Berlin
\newline Berlin, Germany}
\email{raymond@tu-berlin.de}

\date{}

\begin{abstract}
	A class $\mathcal{F}$ of graphs has the induced Erd\H{o}s-P\'osa property if there exists a function $f$ such that 
	for every graph $G$ and every positive integer $k$, $G$ contains either $k$ pairwise vertex-disjoint induced subgraphs that belong to $\mathcal{F}$, or a vertex set of size at most $f(k)$ hitting all induced copies of graphs in $\mathcal{F}$. Kim and Kwon (SODA'18) showed that for a cycle $C_{\ell}$ of length $\ell$,  the class of $C_{\ell}$-subdivisions has the induced Erd\H{o}s-P\'osa property if and only if $\ell\le 4$. 
	In this paper, we investigate whether or not the class of $H$-subdivisions has the induced Erd\H{o}s-P\'osa property for other graphs $H$.

        We completely settle the case when $H$ is a forest or a complete bipartite graph.
        Regarding the general case, we identify necessary conditions on $H$ for the class of $H$-subdivisions to have the induced Erd\H{o}s-P\'osa property.
	For this, we provide three basic constructions that are useful to prove that the class of the subdivisions of a graph does not have the induced Erd\H{o}s-P\'osa property.
	Among remaining graphs, we prove that if $H$ is either the diamond, the $1$-pan, or the $2$-pan, then 
	the class of $H$-subdivisions has the induced Erd\H{o}s-P\'osa property.
\end{abstract}

\maketitle

\section{Introduction}

All graphs in this paper are finite and without loops or parallel edges. In this paper we are concerned with the induced version of the \ep{} property.
This property expresses a duality between invariants of packing and covering related to a class of graphs. Its name originates from the following result.
\begin{theorem}[Erdős-Pósa Theorem, \cite{EP62}]\label{th:ep}
  There is a function $f(k)= \mathcal{O}(k \log k)$ such that 
  for every graph $G$ and every positive integer $k$, 
  $G$ contains either $k$ vertex-disjoint cycles, or a vertex set $X$ of size at most $f(k)$ such that $G - X$ has no cycle.
\end{theorem}
In general, we say that a class of graphs has the \emph{\ep{} property} if a similar statement holds: either we can find in a graph many occurrences of members of the class, or we hit them all with a small number of vertices.
Since the proof of \autoref{th:ep} by Paul Erdős and Lajós Pósa, the line of research of identifying new classes that have the \ep{} property has been very active (see surveys \cite{Reed97tree, Raymond2017recent}). These results are not only interesting because they express some duality between two parameters: they can also be used to design algorithms (see e.g. \cite{JGT:JGT3190120111, fomin2016hitting, Chatzidimitriou2017logopt}).

Several authors attempted to extend \autoref{th:ep} in various directions. One of them is to consider long cycles, i.e.\ cycles of length at least $\ell$ for some fixed integer $\ell\ge 3$.
\begin{theorem}[\cite{MOUSSET201721}, see also \cite{Birmele2007, JGT:JGT21776}]\label{th:longcycles}
  There is a function $f(k, \ell)= \mathcal{O}(k\ell + k \log k)$ such that
  for every graph $G$ and every positive integer $k$, $G$ contains either $k$ vertex-disjoint cycles of length at least $\ell$, or 
  a vertex set $X$ of size at most $f(k, \ell)$ such that $G - X$ has no such cycle.
\end{theorem}

As every cycle contains an induced cycle, \autoref{th:ep} also holds if one replaces \emph{cycle} with \emph{induced cycle} in its statement. This is not so clear with \autoref{th:longcycles} since a long cycle in a graph does not always contain a long induced cycle.
In \cite{jansen2017approximation}, Jansen and Ma. Pilipczuk asked whether \autoref{th:longcycles} holds for induced cycles of length at least $\ell$ for $\ell=4$. This was recently proved to be true by Kim and the first author:

\begin{theorem}[\cite{Kim2017}]\label{th:kim17}
  There is a function $f(k) = \mathcal{O}(k^2 \log k)$ such that, for every graph $G$ and every positive integer $k$, either $G$ has $k$ vertex-disjoint induced cycles of length at least 4, or it contains a vertex set $X$ of size at most $f(k)$ such that $G - X$ has no such cycle.
\end{theorem}
They also showed that the Erd\H{o}s-P\'osa type statement in \autoref{th:kim17} cannot be extended to induced cycles of length at least $\ell$  for fixed $\ell>4$ (even with a different order of magnitude for $f$).

The aim of this paper is to investigate if a statement as that of \autoref{th:kim17} holds for other induced structures.
In order to present it formally, we introduce some notions.
	A class $\mathcal{F}$ of graphs has the \emph{induced Erd\H{o}s-P\'osa property} if there exists a \emph{bounding function} $f\colon \N \to \R$ such that 
	for every graph $G$ and every positive integer $k$, $G$ contains either $k$ pairwise vertex-disjoint induced subgraphs that belong to $\mathcal{F}$, or a vertex set of size at most $f(k)$ hitting all induced copies of graphs in $\mathcal{F}$.\footnote{We decided to use this terminology because the classic \ep{} property considers the subgraph relation as a containment relation and therefore the sentence ``induced $H$-subdivisions have the \ep{} property'' might be confusing.}

A \emph{subdivision of $H$} (\emph{$H$-subdivision} for short) is a graph obtained from $H$ by subdividing some of its edges.
A subgraph of a graph $G$ is called an \emph{induced subdivision of $H$} (or \emph{induced $H$-subdivision}) if it is an induced subgraph of $G$ that is a subdivision of~ $H$.
For graphs $H$ and $G$, we denote by $\upnu_H(G)$ the maximum size of a collection of vertex-disjoint induced subdivisions of $H$ in $G$ (called \emph{packing}). We denote by $\uptau_H(G)$ the minimum size of a subset $X\subseteq V(G)$ (called \emph{hitting set}) such that $G - X$ has no induced subdivision of $H$. 
By definition, the class of $H$-subdivisions has the induced \ep{} property if there is a bounding function $f \colon \N \to \R$ such that $\uptau_H(G) \leq f(\upnu_H(G))$ for every graph $G$.
To avoid a long terminology, we allow to say that $H$-subdivisions have the induced \ep{} property.

\autoref{th:kim17} can be reformulated in terms of the induced \ep{} property of $C_4$-subdivisions.
The authors of \cite{Kim2017} noted that it is an interesting topic to investigate the induced \ep{} property of subdivisions of other graphs.

In this paper, we determine whether $H$-subdivisions have the induced version of the \ep{} property or not, for various graphs $H$.
We note that the classic (i.e.\ non-induced) \ep{} property of subdivisions has been investigated before \cite{JGT:JGT3190120111, liu2017packing, bruhn2017k4, MOUSSET201721}.

\subsection*{Our results}
Towards a classification of graphs based on the induced \ep{} property of their subdivisions, we consider several simple extensions of cycles, depicted in \autoref{fig:smgr} (see \autoref{sec:prelim} for a formal definition). 

\begin{figure}[h]
  \centering
  \begin{tikzpicture}[every node/.style = black node, scale = 0.5]
    \draw (30:1) node (a) {} -- (150:1) node (b) {} -- (-90:1) node (c) {} -- cycle;
    \draw (c) -- ++ (-90:1.5) node {} ;
    \draw (0, -4) node[normal] {\small 1-pan};
    \begin{scope}[xshift = 5cm]
      \draw (30:1) node (a) {} -- (150:1) node (b) {} -- (-90:1) node (c) {} -- cycle;
      \draw (c) -- ++ (-90:0.75) node {} -- ++(-90:0.75) node{}; 
      \draw (0, -4) node[normal] {\small 2-pan};
    \end{scope}
    \begin{scope}[xshift = 10.5cm, yshift = -1.5cm]
      \draw (30:1) node (a) {} -- (150:1) node (b) {} -- (-90:1) node (c) {} -- cycle;
      \draw (a) -- ++ (120:1.732) node {} -- (b);
      \draw (0, -2.5) node[normal] {\small Diamond};
    \end{scope}
  \end{tikzpicture}
  \caption{The graphs mentioned in the statement of \autoref{thm:smallgraphs}.}
  \label{fig:smgr}
\end{figure}
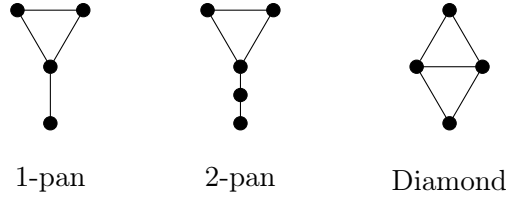

\begin{theorem}\label{thm:smallgraphs}
If $H$ is either the diamond, the $1$-pan, or the $2$-pan, then $H$-subdivisions have the induced \ep{} property with a polynomial bounding function.\label{e:posi}
\end{theorem}
For 1- and 2-pans, we furthermore give a polynomial-time algorithm that constructs a packing or a hitting set of bounded size.
For diamond, we give an algorithm that runs in time $k^{\mathcal{O}(k)}\cdot \abs{G}^c$ for some constant $c$.
This will be explicitly mentioned in the statement of the theorems for each of the graphs.

We then give negative results for graphs satisfying certain general properties.

\begin{theorem}\label{thm:badcases}
Let $H$ be a graph which satisfies one of the following:
\begin{enumerate}
\item $H$ is a forest and two vertices of degree at least 3 lie in the same connected component;\label{badforest}
\item $H$ contains an induced cycle of length at least $5$; \label{badc:c5}
\item $H$ contains a cycle $C$ and two adjacent vertices having no neighbors in $C$; \label{badc:2v}
\item $H$ contains a cycle $C$ and three vertices having no neighbors in $C$; \label{badc:3v}
\item $H=K_{2,n}$ with $n\ge 3$; \label{badc:bip}
\item $H$ is not planar. \label{badc:nonpl}
\end{enumerate}
Then $H$-subdivisions do not have the induced \ep{} property.
\end{theorem}

We remark that if a forest $F$ has no two vertices of degree at least $3$ in the same connected component, 
then $F$ is a disjoint union of subdivisions of stars, and therefore every subdivision of $F$ contains $F$ as an induced subgraph.
Thus, for such a forest $F$, $F$-subdivisions trivially have the induced \ep{} property (see \autoref{l:pathlinear}).
By item \eqref{badforest} of \autoref{thm:badcases}, other forests $F$ will not satisfy this property. 
Therefore, we can focus on graphs containing a cycle.

The rest of \autoref{thm:badcases} provides, for graphs that have cycles, necessary conditions for their subdivisions to have the induced \ep{} property.
By combining together Theorems~\ref{th:kim17}, \ref{thm:smallgraphs}, and \ref{thm:badcases}, we obtain the following dichotomies.

\begin{corollary}\label{thm:dichotomy}\mbox{~}
  \begin{enumerate}
  \item Let $F$ be a forest. Then $F$-subdivisions have the induced \ep{} property if and only if each connected component of $F$ has at most one vertex of degree more than 2.\label{enum:dich-forest}
  \item Let $n,m$ be positive integers with $n\leq m$. Then $K_{n,m}$-subdivisions have the induced \ep{} property if and only if $n \leq 1$ or $m \leq 2$. \label{enum:dich-bip}
  \item Let $n$ be a positive integer. Then the subdivisions of the $n$-pan have the induced \ep{} property if and only if $n \leq 2$. \label{enum:dich-pan}
  \end{enumerate}
\end{corollary}

We also show that if $H$ has a cycle and no vertex of degree more than 3, then there is no $o(k \log k)$ bounding function for $H$ (\autoref{thm:omegaklogk}).

\subsection*{Our techniques}
To obtain negative results, we describe three constructions. Their common point is that any induced subdivision of the considered graph $H$ that they contain has a strictly constrained position. This will ensure that two distinct induced $H$-subdivisions meet. On the other hand, they contain several distinct induced $H$-subdivisions, so one needs many vertices to hit them all. Compared to their non-induced counterparts, induced subdivisions of a fixed graph do not appear in very dense subgraphs such as cliques, which leaves more freedom in the design of our constructions.

The proofs of the positive results for $1$-pan and $2$-pan start similarly.
We first show that if $S$ is a smallest 1-pan- (resp. 2-pan-) subdivision in $G$ that hits all other induced subdivisions of the $1$-pan (resp. 2-pan),
then either $G$ contains $k$ pairwise vertex-disjoint induced subdivisions of the 1-pan, or all the 1-pan-(resp. 2-pan-)subdivisions can be hit with $\mathcal{O}(k \log k)$ vertices. By applying inductively this result we can conclude that the subdivisions of the 1-pan (resp.\ 2-pan) have the induced \ep{} property with gap $\mathcal{O}(k^2 \log k)$. The proof for the diamond is much more involved.

\subsection*{Organization of the paper}
We introduce the definitions and some of the external results that we need in \autoref{sec:prelim}.
In sections from \ref{sec:triwa} to \ref{sec:lowbo}, we present general constructions and use them to obtain the negative results contained in this paper. The positive results about pans and diamonds are proved in Sections \ref{sec:panpan} and \ref{sec:diamond}, respectively.
We conclude in \autoref{sec:concl} with remarks and possible directions for future research. The sections from \ref{sec:triwa} to \ref{sec:diamond} are mostly independent and can be read separately.

\section{Preliminaries}
\label{sec:prelim}

\subsection*{Basics} We denote by $\R$ the set of reals and by $\N$ the set of non-negative integers. 
For an integer $p\ge 1$, we denote by $\N_{\ge p}$ the set of integers greater than or equal to~$p$.
Let $G$ be a graph. We denote by $V(G)$ and $E(G)$ its vertex set and edge set, respectively. 
For every $X\subseteq V(G)$ (resp.\ $X\subseteq E(G)$), we denote by $G-X$ the graph obtained by removing all vertices (resp.\ edges) of $X$ from~$G$. 
For every $S\subseteq V(G)$, we denote by $G[S]$ the subgraph of $G$ induced by $S$.
We use $|G|$ as a shorthand for~$|V(G)|$.

For two graphs $G$ and $H$,  $G\cup H$ is the graph with vertex set $V(G)\cup V(H)$ and edge set $E(G)\cup E(H)$.

For a subgraph $H$ of a graph $G$ and $v\in V(G)\setminus V(H)$, we say that $v$ \emph{dominates} $H$ 
if $v$ is adjacent to every vertex of $H$. In this case, we also say that $v$ is a $H$-dominating vertex.

For $v,w\in V(G)$, we denote by $\dist_G(v,w)$ the number of edges of a shortest path from $v$ to $w$ in $G$; that is the \emph{distance} between $v$ and $w$ in $G$.
This notation is extended to sets $V,W \subseteq V(G)$ as $\dist_G(V,W) = \min_{(v,w)\in V\times W} \dist_G(v,w)$.
When $V$ consists of one vertex $v$, we write $\dist_G(v, W)$ for $\dist_G(\{v\}, W)$.
Let $r \in \N_{\ge 1}$. The \emph{$r$-neighborhood} of a subset $S\subseteq V(G)$, that we denote by $N^r_G[S]$,  is the set of all vertices $w$ such that $\dist_G(w, S)\le r$. 

\subsection*{Subdivisions}
We recall the definition of subdivision given in the introduction. 
The operation of \emph{subdividing} an edge $e$ of a graph $G$ adds a new vertex of degree 2 adjacent to the endpoints of $e$ and removes $e$.
We say that a graph $H'$ is a subdivision of a graph $H$ if $H'$ can be obtained from $H$ by a sequence of edge subdivisions.
We also say that $G$ has an \emph{induced subdivision} of $H$, or shortly an induced $H$-subdivision, if there is an induced subgraph of $G$ that is isomorphic to an $H$-subdivision.

\subsection*{Models}
We formulate an induced subdivision of $H$ in a graph $G$ as a mapping from $H$ to $G$.
A \emph{model} of $H$ in $G$ is an injective function $\varphi$ with domain $V(H) \cup E(H)$ that maps vertices of $H$ to vertices of $G$ and edges of $H$ to induced paths of $G$ such that:
\begin{itemize}
\item for every edge $uv \in E(H)$, $\varphi(u)$ and $\varphi(v)$ are the endpoints of $\varphi(uv)$;
\item for distinct edges $u_1v_1, u_2v_2\in E(H)$, there is no edge between $V(\varphi(u_1v_1)) \setminus \{\varphi(u_1), \varphi(v_1)\}$ and $V(\varphi(u_2v_2)) \setminus \{\varphi(u_2), \varphi(v_2)\}$;
\item for every edge $uv\in E(H)$ and a vertex $w\in V(H)\setminus \{u,v\}$, $\varphi(w)$ has no neighbors in $V(\varphi(uv)) \setminus \{\varphi(u), \varphi(v)\}$;
\item for two non-adjacent vertices $u,v\in V(H)$, $\varphi(u)\varphi(v) \notin E(G)$.
\end{itemize}

It is easy to see that a graph $G$ contains an induced subdivision of a graph $H$ if and only if there is a model of $H$ in $G$. We will often used this fact implicitly.
If $H'$ is an induced subgraph of $H$, then we denote by $\varphi(H')$ the subgraph of $G$ induced by $\bigcup_{v \in V(H')} V(\varphi(v))$.

\subsection*{Small graphs} 
We denote by $K_n$ the complete graph on $n$ vertices, and by $C_n$ the cycle on $n$ vertices.
The \emph{claw} is the complement of the disjoint union of $K_3$ and $K_1$. The \emph{$n$-pan} is the graph obtained from the disjoint union of a path on $n$ edges and $K_3$ by identifying one endpoint of the path with a vertex in $K_3$. The \emph{diamond} is the graph obtained by removing an edge in $K_4$.

\subsection*{General tools}
\label{sec:tools}
We collect in this subsection general results that will be used in the paper.
The first one is a lemma used by Simonovits \cite{Simonovits67genera} to give a new proof of \autoref{th:ep}.

\begin{lemma}[\cite{Simonovits67genera}, see also {\cite[Lemma~2.3.1]{diestel2010graph}}]\label{t:simonovits}
Let $k$ be a positive integer. If a cubic multigraph has at least $24 k \log k$ vertices, then it contains at least $k$ vertex-disjoint cycles.
Furthermore, such $k$ cycles can be found in time $\mathcal{O}(\abs{G}^3)$.
\end{lemma}

The second one is related to an \ep{} type problem about paths intersecting a fixed set of vertices.
Let $G$ be a graph and $A \subseteq V(G)$. An \emph{$A$-path} in $G$ is a path with both end vertices in $A$ and all internal vertices in $V(G)\setminus A$. 
An \emph{$A$-cycle} is a cycle containing at least one vertex of~$A$.

\begin{theorem}[\cite{Gallai1964}]\label{t:gallaiapath}
Let $G$ be a graph, $A\subseteq V(G)$, and $k$ be a positive integer.
Then one can find in time $\mathcal{O}(k\abs{G}^2)$ either
$k$ vertex-disjoint $A$-paths, or
a vertex set $X$ of $G$ with $\abs{X}\le 2k-2$ such that  $G- X$ has no $A$-paths.
\end{theorem}

We use the regular partition lemma introduced in \cite{Choi2017}.
For a sequence $(A_1, \ldots, A_{\ell})$ of finite subsets of an interval $I\subseteq \mathbb{R}$, a partition $\{I_1, \ldots, I_k\}$ of $I$ into intervals 
is called a \emph{regular partition} of $I$ with respect to $(A_1, \ldots, A_{\ell})$ if 
each $i\in\{1,\ldots,k\}$ satisfies one of the following:
\begin{enumerate}[({RP}1)]
\item \label{it:rp1} $A_1\cap I_i=A_2\cap I_i=\cdots = A_{\ell}\cap I_i\not= \emptyset$.
\item $\abs{A_1\cap I_i}=\abs{A_2\cap I_i}=\cdots = \abs{A_{\ell}\cap I_i}>0$, and for all $j,j'\in \{1, \ldots, \ell\}$ with $j<j'$, $\max(A_j\cap I_i)<\min (A_{j'}\cap I_i)$.
\item $\abs{A_1\cap I_i}=\abs{A_2\cap I_i}=\cdots = \abs{A_{\ell}\cap I_i}>0$, and for all $j, j'\in \{1, \ldots, \ell\}$ with $j<j'$, $\max (A_{j'}\cap I_i)< \min (A_j\cap I_i)$. 
\end{enumerate}
The number of parts $k$ is called the \emph{order} of the regular partition.
The following can be obtained using multiple applications of the Erd\H{o}s-Szekeres Theorem~\cite{Erdos1987}.

\begin{lemma}[Regular partition lemma~\cite{Choi2017}]\label{prop:regularpartition}
Let $I\subseteq \mathbb R$ be an interval.  
There exists a function $N$ such that for all positive integers $n, \ell$, the value $N = N(n, \ell)$ satisfies the following.
For every sequence $(A_1, \ldots, A_N)$ of $n$-element subsets of $I$, 
there exist a subsequence $(A_{j_1}, \ldots, A_{j_{\ell}})$ of $(A_1, \ldots, A_N)$ and a regular partition of $I$ with respect to $(A_{j_1}, \ldots, A_{j_{\ell}})$ that has order at most $n$. 
\end{lemma} 

We will use this lemma with $n\in \{2, 3, 4\}$. We note that for fixed $n$, the function $N(n,\ell)$ is a polynomial function in $\ell$, but the known upper bound following from the proof of the result is big.
For instance, $N(2,\ell)=\mathcal{O}(\ell^{40})$ and $N(3,\ell)=\mathcal{O}(\ell^{968})$.

\section{Negative results using the triangle-wall}
\label{sec:triwa}

We prove in this section items \eqref{badforest} and \eqref{badc:bip} of \autoref{thm:badcases}.
The cornerstone of both proofs is the use of a triangle-wall, which we describe hereafter. 
These results are then used to show the dichotomies \eqref{enum:dich-forest} and \eqref{enum:dich-bip} of \autoref{thm:dichotomy}.

Let $n \in \N_{\geq 2}$. The \emph{$n$-garland} is the graph obtained from the path $x_0y_0\dots  x_i y_i \dots x_{n-1} y_{n-1}$ by adding, for every $i \in \intv{0}{n-1}$, a new vertex $z_i$ adjacent to $x_i$ and $y_i$. Let $Q_0, \dots, Q_{n-1}$ be $n$ copies of the $2n$-garland. For every $i \in \intv{0}{n-1}$ and $j\in \intv{0}{2n-1}$, we respectively  denote by $x_j^i, y_j^i, z_j^i$ the copies of the vertices $x_j, y_j, z_j$ in $Q_i$. The graph $\Gamma_n$ (informally called \emph{triangle-wall}) is constructed from the disjoint union of $Q_0, \dots, Q_{n-1}$ as follows:
\begin{itemize}
\item for every even $i \in \intv{0}{n-1}$ and every odd $j \in \intv{0}{2n-1}$, we add an edge between the vertex $z^i_j$ and the vertex $z^{i+1}_j$;
\item for every odd $i \in \intv{0}{n-1}$ and every even $j \in \intv{0}{2n-1}$, we add an edge between the vertex $z^i_j$ and the vertex $z^{i+1}_j$.
\end{itemize}

\begin{figure}[h]
  \centering
  \begin{tikzpicture}
    \begin{scope}[every node/.style = black node, scale = 0.5]
      \foreach \evenrow in {0,2}{
        \begin{scope}[yshift = -2.732*\evenrow cm]
          \foreach \i/\l in {1/1, 5/3, 9/5, 13/7}{
            \draw (\i, 0) node (x{\l,\evenrow}) {} -- ++(1,0) node (y{\l,\evenrow}) {} -- ++(120:1) node (z{\l,\evenrow}) {} -- cycle;
          }
          \foreach \i/\l in {3/2, 7/4, 11/6, 15/8}{
            \draw (\i, 0) node (x{\l,\evenrow}) {} -- ++(1,0) node (y{\l,\evenrow}) {} -- ++(-120:1) node (z{\l,\evenrow}) {} -- cycle;
          }
          \foreach \i in {1,...,7}{
            \pgfmathparse{int(\i+1)}
            \draw (y{\i,\evenrow}) -- (x{\pgfmathresult,\evenrow});
          }
        \end{scope}
      }
      \foreach \oddrow in {1,3}{
        \begin{scope}[yshift = -2.732*\oddrow cm, yscale = -1]
          \foreach \i/\l in {1/1, 5/3, 9/5, 13/7}{
            \draw (\i, 0) node (x{\l,\oddrow}) {} -- ++(1,0) node (y{\l,\oddrow}) {} -- ++(120:1) node (z{\l,\oddrow}) {} -- cycle;
          }
          \foreach \i/\l in {3/2, 7/4, 11/6, 15/8}{
            \draw (\i, 0) node (x{\l,\oddrow}) {} -- ++(1,0) node (y{\l,\oddrow}) {} -- ++(-120:1) node (z{\l,\oddrow}) {} -- cycle;
          }
          \foreach \i in {1,...,7}{
            \pgfmathparse{int(\i+1)}
            \draw (y{\i,\oddrow}) -- (x{\pgfmathresult,\oddrow});
          }
        \end{scope}
      }

      \foreach \evenrow in {0,2}{
        \foreach \i in {2, 4,...,8}{
          \pgfmathparse{int(\evenrow+1)}
          \draw (z{\i,\evenrow}) -- (z{\i,\pgfmathresult});
        } 
      }
      \foreach \oddrow in {1}{
        \foreach \i in {1, 3,...,7}{
          \pgfmathparse{int(\oddrow + 1)}
          \draw (z{\i,\oddrow}) -- (z{\i,\pgfmathresult});
        }
      }
    \end{scope}
    \foreach \i in {0,...,3}{
      \pgfmathparse{int(2*\i+1)}
      \draw (z{\pgfmathresult,0}) node[label=120:$a_\i$] {};
    }
    \foreach \i in {0,...,3}{
      \pgfmathparse{int(7-2*\i)}
      \draw (z{\pgfmathresult,3}) node[label=-60:$b_\i$] {};
    }
    \begin{pgfonlayer}{bg}
        \foreach \row / \next in {0,...,3}{
          \foreach \column in {1,2}{
            \fill[green!80!black] (x{\column,\row}) circle (0.15cm);
            \fill[green!80!black] (y{\column,\row}) circle (0.15cm);
            \fill[green!80!black] (z{\column,\row}) circle (0.15cm);
          }
          \draw[green!80!black, line join = miter, line cap = round, line width = 0.15cm]
          (y{1,\row}) -- (x{1,\row}) -- (z{1,\row}) -- (y{1,\row}) -- (x{2,\row}) -- (z{2,\row}) -- (y{2,\row}) -- (x{2,\row});
        }
        \draw[green!80!black, line join = round, line cap = round, line width = 0.15cm] (z{2,0}) -- (z{2,1}) (z{1,1}) -- (z{1,2}) (z{2,2}) -- (z{2,3});
        \end{pgfonlayer}
  \end{tikzpicture}
  
  \caption{The graph $\Gamma_4$, with $C_0$ depicted in green.}
  \label{fig:gamma4}
\end{figure}
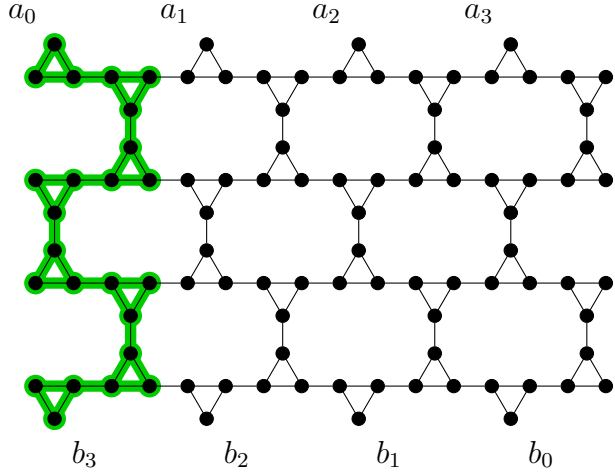

For $i\in \intv{0}{n-1}$, we set
\begin{align*}
  a_i &= z^0_{2i}, \qquad
   b_i =
    \begin{cases}
      z^p_{2(n-i-1)}&\text{if $n$ is even},\\
      z^p_{2(n-i)-1} &\text{otherwise}.
    \end{cases},\\
  \text{and}\ C_i &= \bigcup_{j=0}^{n-1}\{x_{2i}^j, y_{2i}^j, z_{2i}^j, x_{2i + 1}^j, y_{2i + 1}^j, z_{2i + 1}^j\}.                    
\end{align*}
Intuitively, the $Q_i$'s form the rows of $\Gamma_n$ and the $C_i$'s its columns.
The graph $\Gamma_4$ is depicted in \autoref{fig:gamma4}.
Observe that there is no induced claw in $\Gamma_n$.

\subsection{Complete bipartite patterns}
In this section, we complete the study of the induced \ep{} property of subdivisions of complete bipartite graphs by proving item \eqref{enum:dich-bip} of \autoref{thm:dichotomy}.

\begin{lemma}[item \eqref{badc:bip} of \autoref{thm:badcases}]\label{l:bip}
 For every integer $r\ge 3$, the subdivisions of $K_{2, r}$ do not have the induced \ep{} property.
\end{lemma}

\begin{proof}
We construct an infinite family $(G_n)_{n \in \N_{\geq 1}}$ of graphs such that $\upnu_{K_{2,r}}(G_n) = \mathcal{O}(1)$ while $\uptau_{K_{2,r}}(G_n) = \Omega(n)$.

  For $n \in \N_{\geq 1}$, the graph $G_n$ is obtained as follows from the graph $\Gamma_{rn}$:
  \begin{itemize}
  \item for every $i \in \intv{0}{n-1}$, add a new vertex $u_{i}$ and make it adjacent to $a_{ir}, a_{ir+1}, \dots, a_{(i+1)r - 1}$;
  \item for every $i \in \intv{0}{n-1}$, add a new vertex $v_{i}$ and make it adjacent to $b_{ir}, b_{ir+1}, \dots, b_{(i+1)r - 1}$;
  \item for every distinct $i,j \in \intv{0}{n-1}$, add all possible edges between $\{u_i, v_i\}$ and $\{u_j, v_j\}$ (thus, $u_i$ and $v_j$ are not adjacent iff $i=j$).
  \end{itemize}

  Let us show that $\upnu_{K_{2,r}}(G_n) \leq 1$.
  For this we consider a model $\varphi$ of $K_{2,r}$ in $G_n$. Notice that $\varphi(K_{2,r})$ has exactly two vertices of degree $r$, and that they are not adjacent.
  As each of them has $r\geq 3$ pairwise non-adjacent neighbors, we observe that none of these two vertices of degree $r$ belongs to the subgraph $\Gamma_{rn}$ of $G_n$. Furthermore, as they are not adjacent, one is $u_i$ and the other one $v_i$, for some $i\in \intv{0}{n-1}$. 
  
  Let us now focus on the $(u_i,v_i)$-paths of $\varphi(K_{2,r})$.
  Observe there is no edge between the internal vertices of two distinct $(u_i, v_i)$-paths of $\varphi(K_{2,r})$. Therefore if two such paths contain a vertex of $\{u_j, v_j\colon j\in \intv{0}{n-1} \setminus \{i\}\}$ each, one of these vertices is $u_j$ and the other one is $v_j$, for some $j\in \intv{0}{n-1}$. As $r\geq 3$, we deduce that at least one  $(u_i,v_i)$-path of $\varphi(K_{2,r})$ does not contain any vertex of $\{u_j, v_j\colon j\in \intv{0}{n-1} \setminus \{i\}\}$. The set of internal vertices of this path lies in $\Gamma_{rn}$ and connects a neighbor of $u_i$ to a neighbor of $v_i$, i.e.\ a vertex of $\{a_{ir}, a_{ir+1}, \dots, a_{(i+1)r -1}\}$ to a vertex of $\{b_{ir}, b_{ir+1}, \dots, b_{(i+1)r -1}\}$.
  
We just proved that every induced subdivision of $K_{2,r}$ in $G_n$ contains a path of the subgraph $\Gamma_{rn}$ connecting a vertex of $\{a_{ir}, a_{ir+1}, \dots, a_{(i+1)r -1}\}$ to a vertex of $\{b_{ir}, b_{ir+1}, \dots, b_{(i+1)r -1}\}$, for some $i\in \intv{1}{n}$.
As every two such paths meet for different values of $i$, we deduce $\upnu_{K_{2,r}}(G_n) \leq 1$.

We now show that when $n\geq 2r$, $\uptau_{K_{2,r}}(G_n) \geq \frac{n}{2}$.
For this we consider the sets defined for every $i \in \intv{0}{n-1}$ as follows:
\[
  C^+_i = \{u_i, v_{n-i-1}\}  \cup \bigcup_{j=(i + 1)r - 1}^{ir} C_j.
\]
The set $C_i^+$ contains $u_i, v_{n-i-1}$ and the vertices that are, intuitively, in the columns that are between them.
Let $X$ be a subset of $V(G_n)$ with $\ceil{\frac{n}{2}}-1$ vertices. Observe that for distinct $i,j \in \intv{1}{n}$, the sets $C_i^+$ and $C_j^+$ are disjoint.
Hence the $\ceil{n/2}$ sets \[\left \{C^+_i \cup C_{n-i-1}^+,\ i \in \intv{0}{\ceil{\frac{n-1}{2}}} \right \}\] are disjoint. As $|X| = \ceil{\frac{n}{2}} - 1$, we have $X \cap (C^+_i \cup C_{n-i-1}^+) = \emptyset$ for some  $i \in \intv{0}{\ceil{\frac{n-1}{2}}}$.

Recall that $\Gamma_n$ is composed of $n$ disjoint garlands $Q_0$, \dots $Q_{n-1}$ connected together. As $|X|<\frac{n}{2}$ and $r \leq \frac{n}{2}$, some $r$ of these garlands do not intersect $X$. Let $\pi \colon \intv{1}{r} \to \intv{1}{n}$ be an increasing function such that $Q_{\pi(i)} \cap X = \emptyset$, for every $i \in \intv{1}{r}$. Then there is a collection of disjoint paths $P_1, \dots, P_r$ of $G-X$ so that $P_j$ connects $a_{ir+j}$ to $b_{{(i+1)r-1 -j}}$, for every $j \in \intv{1}{r}$. The path $P_j$ can be obtained by following in $G_n[C_{ir +j}]$ (intuitively, the $j$-th column of $C^+_i$) a shortest path from $a_{ir +j}$ to $Q_{\pi(r-1-j)}$, then chordlessly following $Q_{\pi(r-1-j)}$ up to a vertex of $C_{(i+1)r-1 -j}$ and finally following a chordless path to $b_{{(i+1)r-1 -j}}$ in $G_n[C_{(i+1)r-1 -j}]$.
We deduce that $G-X$ contains a model of $K_{2,r}$. As this argument holds for every $X \subseteq V(G)$ such that $|X| \leq \ceil{\frac{n}{2}} - 1$, we deduce that $\uptau_{K_{2,r}}(G_n) \geq \frac{n}{2}$.

Therefore $\upnu_{K_{2,r}}(G_n) = \mathcal{O}(1)$ and $\uptau_{K_{2,r}}(G_n) = \Omega(n)$. This concludes the proof.
\end{proof}

\begin{corollary}[item \eqref{enum:dich-bip} of \autoref{thm:dichotomy}]
  Let $r,r'$ be positive integers with $r \leq r'$. Then the $K_{r,r'}$-subdivisions have the induced \ep{} property if and only if $r \leq 1$ or $r' \leq 2$.
\end{corollary}
	\begin{proof}
	  When $r=1$, the result holds by \autoref{l:pathlinear} (to be proved in the next section). When $r = r' = 2$, $K_{2,r}$ is a cycle on four vertices and the result holds by \autoref{th:kim17}.
	  When $r=2$ and $r'\ge 3$, the result follows from \autoref{l:bip}.
  In the case where $r\geq 3$, then $K_{r,r'}$ is not planar and the result follows from \autoref{l:notplanar}.
	\end{proof}

\subsection{Acyclic patterns}
\label{sec:acyclic}

We show in this section that if a graph is acyclic, then its subdivisions either have the induced \ep{} property with a linear bounding function (\autoref{l:pathlinear}), or they do not have the induced \ep{} property (\autoref{lem:twodegthree}). This proves item \eqref{enum:dich-forest} of \autoref{thm:dichotomy}.

\begin{lemma}\label{l:pathlinear}
Let $H$ be a graph whose connected components are paths or subdivided stars. Then the subdivisions of $H$ have the induced \ep{} property with a bounding function of order~$\mathcal{O}(k)$.
\end{lemma}

\begin{proof}
  We show that for every graph $G$ we have $\uptau_H(G) \leq \upnu_H(G) \cdot |H|$. 
  Towards a contradiction we assume that the above statement has a counterexample $G$, that we choose to have the minimum number of vertices.
  Clearly $\upnu_H(G) \geq 1$. Let $M$ be an induced subdivision of $H$ with minimum number of vertices. As the connected components of $H$ are paths and subdivided stars, $M$ is a copy of $H$.
  Hence $|M| = |H|$. By minimality of $G$, we have $\uptau_H(G- V(M)) \leq \upnu_H(G-V(M))\cdot |H|\leq (\upnu_H(G)-1)\cdot |H|$.
  We deduce
  \begin{align*}
    \uptau_H(G) &\leq \uptau_H(G - V(M)) + |M|\\
    &\leq  (\upnu_H(G)-1)|H| + |H|\\
    &= \upnu_H(G)\cdot |H|.
  \end{align*}
  This contradicts the definition of~$G$.
\end{proof}

\begin{lemma}[item \eqref{badforest} of \autoref{thm:badcases}]\label{lem:twodegthree}
  Let $H$ be a forest, one connected component of which contains at least two vertices of degree at least 3. Then the subdivisions of $H$ do not have the induced \ep{} property.
\end{lemma}

\begin{proof}
  Our goal is to construct an infinite family of graphs $(G_n)_{n \in \N_{\geq 1}}$ such that $\upnu_{H}(G_n) = \mathcal{O}(1)$ while $\uptau_{H}(G_n) = \Omega(n)$.
Let $J$ be a connected component of $H$ which contains at least two vertices of degree at least 3 and let $u$ and $u'$ be two such vertices, that are at minimal distance of each other. Let $\{v,v'\}$ be an edge of the path connecting $u$ to $u'$, with the convention that $u$ is closer to $v$ than to~$v'$. We call $T$ the connected component of $H-\{vv'\}$ that contains $v$, $T'$ that that contains $v'$, and $D$ the union of the remaining connected components. We also call $F$ (resp.\ $F'$) the graph obtained from $T$ (resp.\ $T'$) by removing the vertices of the path from $u$ to $v$ (resp. $u'$ to $v'$). See \autoref{fig:graphH} for an illustration.

  \begin{figure}[h]
    \centering
    \begin{tikzpicture}[rotate = 90]
      \begin{scope}
        \draw[every node/.style = black node] (0,0) node[label = 90:$v$] (v) {} --
        ++(0,0.75) node[normal, fill = white, inner sep = 1] {$\dots$} --
        ++(0,0.75) node[black node, label = 90:$u$] (u) {}
        (u) -- ++(90-70*1.5/4:1) node (u3) {}
        (u)  ++(90-70*0.5/4:1) node[normal, yshift = 0.12cm, rotate = -10,inner sep = 0] {\tiny $\vdots$}
        (u) -- ++(90+70*0.5/4:1) node (u2) {}
        (u) -- ++(90+70*1.5/4:1) node (u1) {};
        \draw (u1) -- ++(90+70*1.5/4-30:0.5) node (v11) {}
        (u1) -- ++(90+70*1.5/4+30:0.5) node (v12) {}
        (v11.center) -- (v12.center);
        \draw (u2) -- ++(90-70*0.5/4-30:0.5) node (v21) {}
        (u2) -- ++(90-70*0.5/4+30:0.5) node (v22) {}
        (v21.center) -- (v22.center);
        \draw (u3) -- ++(90-70*1.5/4-30:0.5) node (v31) {}
        (u3) -- ++(90-70*1.5/4+30:0.5) node (v32) {}
        (v31.center) -- (v32.center);
        \draw[rounded corners, color = gray] (u) ++(0.95, 0.55) -- ++(0, 1) -- ++(-1.9, 0) node[midway, label = 180:$F$] {} -- ++(0, -1) -- cycle;
        \draw[rounded corners, color = gray] (v) ++(1.15, -0.33) -- ++(0, 4.5) -- ++(-2.3, 0) node[midway, label = 180:$T$] {} -- ++(0, -4.5) -- cycle;
      \end{scope}
      \begin{scope}[yscale = -1]
        \draw[every node/.style = black node] (v) --
        ++(0,1) node[label = 90:$v'$] (vp) {} --
        ++(0,0.75) node[normal, fill = white, inner sep = 1] {$\dots$} --
        ++(0,0.75) node[black node, label = 90:$u'$] (u) {}
        (u) -- ++(90-70*1.5/4:1) node (u3) {}
        (u)  ++(90-70*0.5/4:1) node[normal, yshift = 0.12cm, rotate = -10,inner sep = 0] {\tiny $\vdots$}
        (u) -- ++(90+70*0.5/4:1) node (u2) {}
        (u) -- ++(90+70*1.5/4:1) node (u1) {};
        \draw (u1) -- ++(90+70*1.5/4-30:0.5) node (v11) {}
        (u1) -- ++(90+70*1.5/4+30:0.5) node (v12) {}
        (v11.center) -- (v12.center);
        \draw (u2) -- ++(90-70*0.5/4-30:0.5) node (v21) {}
        (u2) -- ++(90-70*0.5/4+30:0.5) node (v22) {}
        (v21.center) -- (v22.center);
        \draw (u3) -- ++(90-70*1.5/4-30:0.5) node (v31) {}
        (u3) -- ++(90-70*1.5/4+30:0.5) node (v32) {}
        (v31.center) -- (v32.center);
        \draw[rounded corners, color = gray] (u) ++(0.95, 0.55) -- ++(0, 1) -- ++(-1.9, 0) node[midway, label = 0:$F'$] {} -- ++(0, -1) -- cycle;
        \draw[rounded corners, color = gray] (vp) ++(1.15, -0.33) -- ++(0, 4.5) -- ++(-2.3, 0) node[midway, label = 0:$T'$] {} -- ++(0, -4.5) -- cycle;
        \draw (0,7) node[white node, minimum size  = 0.75cm] (d) {$D$};
      \end{scope}
    \end{tikzpicture}
    \caption{The forest $H$.}
    \label{fig:graphH}
  \end{figure}
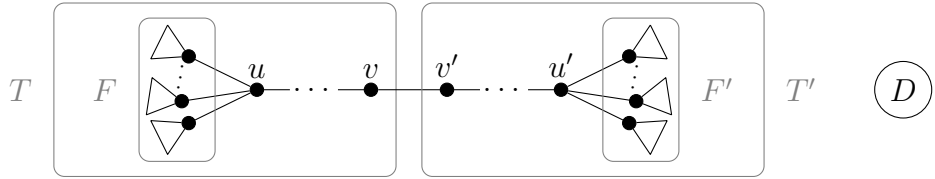
  
  Recall that the triangle-wall $\Gamma_n$ is defined in \autoref{sec:tools} and has special vertices called $a_i$ and~$b_i$.
  For every $n\ \in \N_{\geq 1}$, we construct the graph $G_n$ from a copy of $\Gamma_n$ as follows:
  \begin{enumerate}
  \item for every $i \in \intv{0}{n-1}$, we add a new copy $T_i$ of $T$ and add an edge from the copy of $v$ to $a_i$, and let $F_i$ be the copy of $F$ in $T_i$;
  \item for every $i \in \intv{0}{n-1}$, we add a new copy $T'_i$ of $T'$ and add an edge from the copy of $v'$ to $b_i$, and let $F_i'$ be the copy of $F$ in $T_i'$;
  \item for every $i \in \intv{0}{n-1}$, we add a new copy $D_i$ of $D$;
  \item for every distinct $i,j \in \intv{0}{n-1}$, we add all edges between $F_{i} \cup D_i$ and $F_{j} \cup D_j$, between $F'_i \cup D_i$ and $F'_j \cup D_j$ and between $F_i\cup D_i$ and $F'_j \cup D_j$.\label{e:compl}
  \end{enumerate}

\begin{figure}[h]
  \centering
  \begin{tikzpicture}[rotate = 90]
    \begin{scope}[every node/.style = black node, scale = 0.5]
      \foreach \evenrow in {0,2}{
        \begin{scope}[yshift = -2.732*\evenrow cm]
          \foreach \i/\l in {1/1, 5/3, 9/5, 13/7}{
            \draw (\i, 0) node (x{\l,\evenrow}) {} -- ++(1,0) node (y{\l,\evenrow}) {} -- ++(120:1) node (z{\l,\evenrow}) {} -- cycle;
          }
          \foreach \i/\l in {3/2, 7/4, 11/6, 15/8}{
            \draw (\i, 0) node (x{\l,\evenrow}) {} -- ++(1,0) node (y{\l,\evenrow}) {} -- ++(-120:1) node (z{\l,\evenrow}) {} -- cycle;
          }
          \foreach \i in {1,...,7}{
            \pgfmathparse{int(\i+1)}
            \draw (y{\i,\evenrow}) -- (x{\pgfmathresult,\evenrow});
          }
        \end{scope}
      }
      \foreach \oddrow in {1,3}{
        \begin{scope}[yshift = -2.732*\oddrow cm, yscale = -1]
          \foreach \i/\l in {1/1, 5/3, 9/5, 13/7}{
            \draw (\i, 0) node (x{\l,\oddrow}) {} -- ++(1,0) node (y{\l,\oddrow}) {} -- ++(120:1) node (z{\l,\oddrow}) {} -- cycle;
          }
          \foreach \i/\l in {3/2, 7/4, 11/6, 15/8}{
            \draw (\i, 0) node (x{\l,\oddrow}) {} -- ++(1,0) node (y{\l,\oddrow}) {} -- ++(-120:1) node (z{\l,\oddrow}) {} -- cycle;
          }
          \foreach \i in {1,...,7}{
            \pgfmathparse{int(\i+1)}
            \draw (y{\i,\oddrow}) -- (x{\pgfmathresult,\oddrow});
          }
        \end{scope}
      }

      \foreach \evenrow in {0,2}{
        \foreach \i in {2, 4,...,8}{
          \pgfmathparse{int(\evenrow+1)}
          \draw (z{\i,\evenrow}) -- (z{\i,\pgfmathresult});
        } 
      }
      \foreach \oddrow in {1}{
        \foreach \i in {1, 3,...,7}{
          \pgfmathparse{int(\oddrow + 1)}
          \draw (z{\i,\oddrow}) -- (z{\i,\pgfmathresult});
        }
      }
      \begin{pgfonlayer}{bg}    
        \foreach \row / \next in {0,...,3}{
          \foreach \column in {1,2}{
            \fill[green!80!black] (x{\column,\row}) circle (0.3cm);
            \fill[green!80!black] (y{\column,\row}) circle (0.3cm);
            \fill[green!80!black] (z{\column,\row}) circle (0.3cm);
          }
          \draw[green!80!black, line join = miter, line cap = round, line width = 0.15cm]
          (y{1,\row}) -- (x{1,\row}) -- (z{1,\row}) -- (y{1,\row}) -- (x{2,\row}) -- (z{2,\row}) -- (y{2,\row}) -- (x{2,\row});
        }
        \draw[green!80!black, line join = round, line cap = round, line width = 0.15cm] (z{2,0}) -- (z{2,1}) (z{1,1}) -- (z{1,2}) (z{2,2}) -- (z{2,3})
        (z{1,0}) -- ++(0,3) node (u) {};
        \fill[green!80!black] (u) circle (0.3cm);
        \draw[scale =2, green!80!black, line join = round, line cap = round, line width = 0.15cm]
        (u) -- ++(90-70*1.5/4:1) node[draw = none] (u3) {}
        (u) -- ++(90+70*0.5/4:1) node[draw = none] (u2) {}
        (u) -- ++(90+70*1.5/4:1) node[draw = none] (u1) {};
        \fill[green!80!black] (u1) circle (0.3cm);
        \fill[green!80!black] (u2) circle (0.3cm);
        \fill[green!80!black] (u3) circle (0.3cm);
        \fill[green!80!black, rotate = -60, shift = {(u1)}] (0,0) -- ++(150:1) -- ++(-90:1) -- cycle;
        \fill[green!80!black, rotate = -100, shift = {(u2)}] (0,0) -- ++(150:1) -- ++(-90:1) -- cycle;
        \fill[green!80!black, rotate = -120, shift = {(u3)}] (0,0) -- ++(150:1) -- ++(-90:1) -- cycle;
        \draw[green!80!black, line join = round, line cap = round, line width = 0.15cm]
        (z{1,3}) -- ++(0,-3) node (u) {};
        \fill[green!80!black] (u) circle (0.3cm);
        \begin{scope}[yscale = -1]
        \draw[scale =2, green!80!black, line join = round, line cap = round, line width = 0.15cm]
        (u) -- ++(90-70*1.5/4:1) node[draw = none] (u3) {}
        (u) -- ++(90+70*0.5/4:1) node[draw = none] (u2) {}
        (u) -- ++(90+70*1.5/4:1) node[draw = none] (u1) {};
        \fill[green!80!black] (u1) circle (0.3cm);
        \fill[green!80!black] (u2) circle (0.3cm);
        \fill[green!80!black] (u3) circle (0.3cm);
        \fill[green!80!black, rotate = -60, shift = {(u1)}] (0,0) -- ++(150:1) -- ++(-90:1) -- cycle;
        \fill[green!80!black, rotate = -100, shift = {(u2)}] (0,0) -- ++(150:1) -- ++(-90:1) -- cycle;
        \fill[green!80!black, rotate = -120, shift = {(u3)}] (0,0) -- ++(150:1) -- ++(-90:1) -- cycle;
      \end{scope}
      \end{pgfonlayer}

    \end{scope}
    \foreach \i/\l in {1/0, 2/1, 3/2,4/3}{
      \pgfmathparse{int(2*\i-1)}
      \draw (z{\pgfmathresult,0}) node[label=90:$a_\l$] (a) {};
      \begin{scope}
        \draw[every node/.style = black node] (a)
        ++(0,0.75) node[normal, fill = none, inner sep = 1] (midn) {$\dots$}
        ++(0,0.75) node[black node, label = 90:$u_\l$] (u) {}
        (a) -- (midn.east) (midn.west) -- (u)
        (u) -- ++(90-70*1.5/4:1) node (u3) {}
        (u)  ++(90-70*0.5/4:1) node[normal, yshift = 0.12cm, rotate = -10,inner sep = 0] {\tiny $\vdots$}
        (u) -- ++(90+70*0.5/4:1) node (u2) {}
        (u) -- ++(90+70*1.5/4:1) node (u1) {};
        \draw[rotate = -60, shift = {(u1)}] (0,0) -- ++(150:0.5) -- ++(-90:0.5) -- cycle;
        \draw[rotate = -100, shift = {(u2)}] (0,0) -- ++(150:0.5) -- ++(-90:0.5) -- cycle;
        \draw[rotate = -120, shift = {(u3)}] (0,0) -- ++(150:0.5) -- ++(-90:0.5) -- cycle;
        \draw[rounded corners, color = gray] (u) ++(0.95, 0.55) -- ++(0, 1) -- ++(-1.9, 0) -- ++(0, -1) -- cycle
        (u) ++(0, 2) node {$F_\l$};
    \end{scope}
    }
    \foreach \i/\l in {1/0, 2/1, 3/2,4/3}{
      \pgfmathparse{int(7-2*\i +2)}
      \draw (z{\pgfmathresult,3}) node[label=90:$b_{\l}$] (b) {};
      \begin{scope}[yscale = -1]
        \draw[every node/.style = black node] (b)
        ++(0,0.75) node[normal, fill = none, inner sep = 1] (midn) {$\dots$}
        ++(0,0.75) node[black node, label = 90:$u'_{\l}$] (u) {}
        (b) -- (midn.west) (midn.east) -- (u)
        (u) -- ++(90-70*1.5/4:1) node (u3) {}
        (u) -- ++(90+70*0.5/4:1) node (u2) {}
        (u) -- ++(90+70*1.5/4:1) node (u1) {};
        \draw[rotate = -60, shift = {(u1)}] (0,0) -- ++(150:0.5) -- ++(-90:0.5) -- cycle;
        \draw[rotate = -100, shift = {(u2)}] (0,0) -- ++(150:0.5) -- ++(-90:0.5) -- cycle;
        \draw[rotate = -120, shift = {(u3)}] (0,0) -- ++(150:0.5) -- ++(-90:0.5) -- cycle;
        \draw[rounded corners, color = gray] (u) ++(0.95, 0.55) -- ++(0, 1) -- ++(-1.9, 0) -- ++(0, -1) -- cycle
        (u) ++(0, 2) node {$F'_{\l}$};
      \end{scope}
    }
  \end{tikzpicture}
  
  \caption{The graph $G_4$. The $D_i$'s and the edges between the $V(F_i) \cup V(F'_i)$'s (framed) are not depicted.}
  \label{fig:g4}
\end{figure}
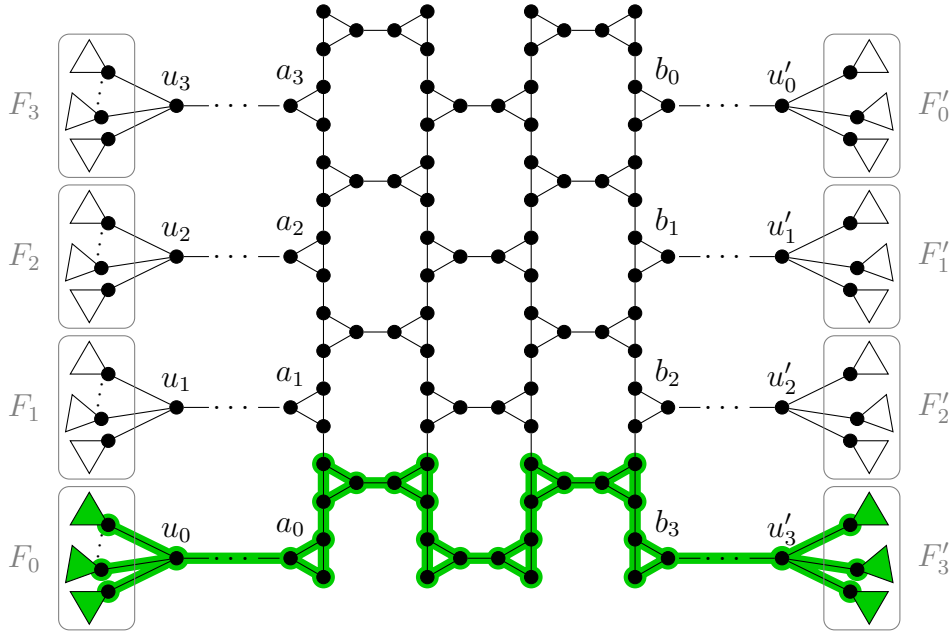

  Let $n \in \N$ and let us show that every induced subdivision of $H$ in $G_n$ has a very restricted position. For convenience we refer to the copy of $\Gamma_n$ in $G_n$ as~$\Gamma_n$ and if $w$ is a vertex of $H$, we refer by $w_i$ to its copy in $T_i$, $T_i'$, or $D_i$. We denote by $P_{xy}$ the unique path of $T$ (resp.\ $T'$) with endpoints $x,y \in V(T)$ (resp.\ $x,y \in V(T')$) and similarly for $P_{x_iy_i}$ in $T_i$ and~$T'_i$.
  
  Let $\varphi$ be a model of $H$ in $G_n$.

  \begin{claim}\label{c:deg3}
    There is an integer $i\in \intv{0}{n-1}$ such that every vertex of degree at least 3 of $\varphi(H)$ belongs to $V(T_i) \cup V(T_i') \cup V(D_i)$.
  \end{claim}
  \begin{proofclaim}
    Let $x,y$ be two vertices of degree 3 in $\varphi(H)$. As mentioned above, these vertices do not belong to $\Gamma_n$. Let us assume towards a contradiction that for distinct $i,j \in\intv{0}{n-1}$ we have
    \[
      x \in V(T_i) \cup V(T_i') \cup V(D_i)\quad \text{and} \quad y \in V(T_j) \cup V(T_j') \cup V(D_j).
    \]
    The first case we consider is when $x \in \{u_i, u'_i\}$. Then at most one of its neighbors (in $\varphi(H)$) belongs to $P_{u_iv_i}$ or $P_{u_i'v_i'}$ so $x$ has two neighbors $x_1$ and $x_2$ that belong to $V(F_i) \cup V(F_i') \cup V(D_i)$.
    If additionally $y \in \{u_j, u'_j\}$, then we can similarly deduce that it has (in $\varphi(H)$) two neighbors that belong to $V(F_j) \cup V(F_j')$. By construction they are both adjacent to $x_1$ and $x_2$ and form together with them an induced $C_4$, a contradiction. We deduce $y \notin \{u_j, u'_j\}$, i.e. $y \in V(F_j) \cup V(F_j')$. But then $y$ is adjacent to $x_1$ and $x_2$, which again forms an induced $C_4$ and hence is not possible.
    Notice that the case where  $x \notin \{u_i, u'_i\}$ and $y \in \{u_j, u'_j\}$ is symmetric.
    
    The only remaining case to consider is then when $x \notin \{u_i, u'_i\}$ and $y \notin \{u_j, u'_j\}$. Then $xy \in E(\varphi(H))$.
    At most one neighbor of $x$ (in $\varphi(H)$) is $u_i$, so $x$ has a neighbor $x_1\neq y$ in $V(F_k) \cup V(F_k') \cup V(D_k)$ for some $k\in \intv{1}{n}$. As this neighbor cannot be adjacent to $y$ without creating a cycle, it belongs in fact to $V(F_j) \cup V(F_j') \cup V(D_j)$. Symmetrically, $y$ has a neighbors distinct from $x$ in $V(F_i) \cup V(F_i') \cup V(D_i)$. But then $G_n[\{x,x_1,y,y_1\}]$ contains a cycle, a contradiction.
    This concludes the proof.
  \end{proofclaim}

  Let $i$ be as in the statement of \autoref{c:deg3}. Observe that $\varphi(H)$ contains no vertex of $V(T_j) \cup V(T_j') \cup V(D_j)$ for $j \in \intv{0}{n-1}\setminus\{i\}$, otherwise it would contain a cycle:
  \[
    V(\varphi(H)) \subseteq V(T_i) \cup V(T_i') \cup V(D_i) \cup V(\Gamma_n).
  \]
  By construction $|E(T_i \cup T_i' \cup D_i)| = |E(H)| - 1$. We deduce that $\varphi(H)$ contains a path $Q$ from a vertex $x$ of $T_i$ to a vertex $y$ of $T_i'$. 
  As $Q$ does not intersect $V(T_j) \cup V(T_j') \cup V(D_j)$ for $j \in \intv{0}{n-1}\setminus\{i\}$, a subpath of it links $a_i$ to $b_i$ in $\Gamma_n$.
  
We proved that if there is an induced subdivision of $H$ in $G_n$, it contains an path of $\Gamma_n$ from $a_i$ to $b_i$ for some $i\in \intv{0}{n-1}$.
Notice that two such paths meet for distinct values of~$i$.
Consequently, $\upnu_H(G_n) \leq 1$. On the other hand, $G_p[V(T_i) \cup V(P) \cup V(T'_i) \cup V(D_i)]$ is an induced subdivision of $H$ for every $i \in \intv{1}{n}$ and every chordless path $P$ of $\Gamma_n$ connecting $a_i$ to~$b_i$. Hence $\upnu_H(G_n) = 1$.

We now show that $\uptau_H(G_n) = \Omega(n)$.

\begin{claim}
For every $n\in \N$, $\uptau_H(G_n) \geq \frac{n}{2}$.  
\end{claim}

\begin{proofclaim}
  This proof is similar to the end of the proof of \autoref{l:bip}.
  Let $n \in \N$. We consider a set $X$ of $\ceil{\frac{n}{2}} - 1$ vertices of $G_n$ and show that $G_n - X$ contains a induced subdivision of~$H$.
For every $i \in \intv{0}{n-1}$, we set
\[C_i^+ = V(T_i) \cup V(D_i) \cup V(T'_{p-i-1}) \cup C_i.\]
Intuitively $C_i$ contains the vertices of the copies $T_i$ and $T_{p-i+1}$, a path linking them, and $D_i$; see \autoref{fig:g4} for a depiction of $C_1$ in~$G_4$ (in green).
As in the proof of \autoref{l:bip} we deduce the existence of an integer $i$ such that none of $C_i^+$ and $C_{n-i-1}^+$ intersects $X$.
As $|X|<n$, for some $j\in \intv{0}{n-1}$ the garland $Q_j$ (which is a subgraph of $\Gamma_n$) does not contain a vertex of $X$. Besides, any of these garlands intersect each of $C_0, \dots, C_{n-1}$. We deduce that in $G - X$, $C_{i}$ and $C_{n-i-1}$ (which by definition of $i$ do not intersect $X$) belong to the same connected component. An induced subdivision of $H$ can then be found by connecting the vertex $a_i$ of $T_i$ to the vertex $b_i$ of $T_{n-i-1}$ using a chordless path. As this argument works for every $X \subseteq V(G)$ such that $|X| \leq \ceil{\frac{n}{2}} - 1$, we deduce that $\uptau_H(G_n) \geq \frac{n}{2}$.
\end{proofclaim}
We constructed a family $(G_n)_{n \in \N}$ of graphs such that $\upnu_H(G_n) = \mathcal{O}(1)$ and $\uptau_H(G_n) = \Omega(n)$. This proves that subdivisions of $H$ do not have the induced \ep{} property.
\end{proof}

\section{Negative results using hypergraphs}\label{sec:hypergraph}
\label{sec:neghype}

We prove in this section negative results about the induced \ep{} property of the subdivisions of graphs containing long cycles and of graphs that have an edge away from a cycle, which are respectively items \eqref{badc:c5} and \eqref{badc:2v} of \autoref{thm:badcases}.
We start with an easy lemma.
\begin{lemma}\label{lem:subdivisionlemma}
Let $H$ be a graph containing a cycle $C$, let $H'$ be a subdivision of $H$, and let $e$ be an edge of $H'$ contained in the subdivision of $C$.
Then $H'-\{e\}$ does not contain an induced subdivision of $H$.
\end{lemma}
\begin{proof}  
  For a graph $G$, let $c(G)$ denote the number of distinct (and not disjoint) cycles in $G$.
  Observe that there is a bijection between the cycles of $G$ and those of a subdivision $G'$ of $G$, hence $c(G) = c(G')$.
  Besides, if $G'$ is a subgraph of $G$, then every cycle of $G'$ is a cycle of $G$, hence $c(G') \leq c(G)$.

  As $e$ belongs to a cycle, $c(H'- \{e\}) < c(H') = c(H)$.
  If $H'-\{e\}$ contained a subdivision $H''$ of $H$ we would have the following contradiction:
  \[c(H) = c(H'') \leq c(H'-\{e\}) < c(H') = c(H).\qedhere \]
\end{proof}

Our proofs rely on a suitable modification of the construction given in \cite{Kim2017} to show that 
for $\ell\ge 5$, $C_{\ell}$-subdivisions have no induced \ep{} property.

A pair $(X,E)$ consisting of a set $X$ and a family $E$ of non-empty subsets of $X$
is called a \emph{hypergraph}.
Each element in $E$ is called a \emph{hyperedge}, and 
for a hypergraph $H$, let $E(H)$ denote the set of hyperedges in $H$. 
A subset $Y$ of $X$ is called a \emph{hitting set} if for every $F\in E$, $Y\cap F\neq \emptyset$.
For positive integers $a,b$ with $a\geq b$, 
let $U_{a,b}$ be the hypergraph $(X, E)$ such that 
$X=\intv{1}{a}$ and $E$ is the set of all subsets of $X$ of size $b$.
It is not hard to observe that for every positive integer $k$, every two hyperedges of $U_{2k-1, k}$ intersect and that
the minimum size of a hitting set of $U_{2k-1, k}$ is precisely~$k$.

The bipartite graph $UB_k$ with bipartition $(A,B)$ is defined as follows:
\begin{itemize}
\item $A=\{1, \ldots, 2k-1\}$;
\item for every hyperedge $F=\{a_1, a_2, \ldots, a_k\}$ of $U_{2k-1, k}$ with $a_1< \cdots <a_k$, we add fresh vertices $p_1^F, \dots, p^F_{k+1}$ to $B$ and the edges of the path $P_F=p^F_1a_1p^F_2a_2 \cdots p^F_ka_kp^F_{k+1}$.
\end{itemize}
We also set $\mathcal{P}(UB_k)=\{P_F:F\in E(U_{2k-1, k})\}$.

We are ready to prove one of the main results of this section.
\begin{theorem}[item \eqref{badc:c5} of \autoref{thm:badcases}]\label{thm:cyclelength5}
If a graph $H$ contains an induced cycle of length at least $5$, then 
the subdivisions of $H$ have no induced \ep{} property.
\end{theorem}

\begin{proof}
Let $H$ be a graph containing an induced cycle $C$ of length at least $5$, and let $m:=\abs{H}$.
We construct an infinite family $(G_n)_{n\in \mathbb{N}_{\ge m}}$ of graphs such that 
$\upnu_{H}(G_n) = 1$ while $\uptau_{H}(G_n) \geq n$ for every $n \in \N_{\geq m}$.

Let $n\ge m$ be an integer and let us fix an edge $uv$ of~$C$.
Let $H'$ be the graph obtained from $H$ by subdividing each edge into a path of length $2n$. In particular, we subdivide $uv$ into $uw_1w_2 \cdots w_{2n-1}v$.
We construct a graph $G=G_n$ from $UB_n$ as follows:
\begin{itemize}
\item for each path $P$ of $\mathcal{P}(UB_n)$, we take a copy $H_P$ of $H'$ and identify the copy of $uw_1w_2 \cdots w_{2n-1}w$ with the path $P$, in the same order
\item then we add all possible edges between $V(H_P)\setminus  A$ and $V(H_{P'})\setminus A$ for two distinct paths $P$ and $P'$ in $\mathcal{P}(UB_n)$.
\end{itemize}

Note that $A$ is still an independent set in $G$.
Clearly, each $H_P$ is an induced subdivision of $H$. 

\begin{claim}\label{claim:fq1}
Every induced subdivision of $H$ in $G$ is contained in $V(H_P)\cup A$ for some $P\in \mathcal{P}(UB_n)$, and 
contains the path $P$.
\end{claim}
\begin{proofclaim}

Let $\varphi$ be a model of $H$ in $G$ and $C':=\varphi(C)$.

We first show that $C'$ is contained in $H_P$ for some $P\in \mathcal{P}(UB_n)$.
Since $A$ is an independent set,  $C'$ contains a vertex $v$ of $H_P- A$ for some $P\in \mathcal{P}(UB_n)$.
Suppose for contradiction that $C'$ is not contained in $H_P$.
Then $C'$ also contains a vertex $v'$ of $H_{P'} - A$ for some $P' \in \mathcal{P}(UB_n)$, $P'\neq P$.
By construction $v$ is adjacent to $v'$.

As $C'$ contains no triangle, $C'$ does not contain any vertex in $H_{P''}$ for $P''\in \mathcal{P}(UB_n)\setminus \{P, P'\}$.

We now analyze the order of $C'[A]$. 
Let us observe that if $I$ is an independent set of a cycle $M$, we have $\abs{N_M(I)}\ge \abs{I}$ since we can define an injective mapping from $I$ to $V(M)\setminus I$ by mapping $v\in V(T)$ to its left neighbor.
Since $C'[V(C')\cap A]$ is an independent set in $G$, the previous observation implies that $\abs{C' - A }\ge \abs{C'[V(C')\cap A]}$.

Suppose $C' - (\{v,v'\} \cup A)$ has two vertices $w$ and $w'$.
Then either $C'$ has a vertex with three neighbors in $C'$ (when both $w$ and $w'$ belong to one of $H_P - A$ and $H_{P'} - A$), or it contains an induced $C_4$ (when $\{w,w'\}$ intersects both $V(H_P - A)$ and $V(H_{P'} - A)$). This is not possible as $C'$ is an induced cycle of length at least 5.
We deduce that $C' -A$ has at most three vertices.

Observe that if $\abs{C'-A} =  \abs{V(C')\cap A} = 3$, then $C'$ is a $C_6$ and $C'-A$ is an independent set. This is not possible as $vv'\in E(C'-A)$, hence $|C'|\leq 5$.
Since $V(C')\cap A$ is an independent set in $C'$, we deduce $\abs{V(C')\cap A} \leq 2$. 

Now, if $|V(C')\cap A] = 2$, then $C' - A$ has three vertices and only one edge, contradicting the fact that it intersects both $H_{P} - A$ and $H_{P'} - A$.
Consequently, $C'$ has length 4, which is a contradiction.

We conclude that $C'$ is contained in $V(H_P)\cup A$ for some $P\in \mathcal{P}(UB_{n})$.
Note that every cycle in $H_P$ contains more than $m$ vertices of $H_P - V(P)$, by construction.
Therefore, if $\varphi(H)$ contains a vertex of $H_{P'} - A$ for some $P'\in \mathcal{P}(UB_{n})$, $P'\neq P$, 
then this vertex should have degree at least $m$ in $\varphi(H)$, which cannot happen in an induced subdivision of~$H$.
Thus, all other vertices are also contained in $V(H_P)\cup A$, which proves the first part of the claim.
Since $P$ is a part of the subdivisions of the cycle of $H$ in $H_P$, by \autoref{lem:subdivisionlemma}, the subdivision of $H$ should contain the path $P$.
\end{proofclaim}

Since every induced subdivision of $H$ in $G$ contains a path in $\mathcal{P}(UB_n)$, 
no two induced subdivisions of $H$ in $G$ are vertex-disjoint, hence $\upnu_{H}(G) = 1$. 
For every vertex subset $S$ of size at most $n-1$, 
there exists $P\in \mathcal{P}(G)$ such that $H_P$ contains no vertex of $S$.
Therefore $\uptau_{H}(G) \geq n$, as required.
\end{proof}

Using the same construction, we can also prove the following.

\begin{theorem}[item \eqref{badc:2v} of \autoref{thm:badcases}]\label{thm:twoadjacent}
Let $H$ be a graph containing a cycle $C$ and two adjacent vertices having no neighbors in $C$.
Then the subdivisions of $H$ do not have the induced \ep{} property.
\end{theorem}
\begin{proof}
We use the same construction in \autoref{thm:cyclelength5}.
Similarly, we claim that 
every induced subdivision of $H$ in $G$ is contained in $V(H_P)\cup A$ for some $P\in \mathcal{P}(UB_n)$.

Let $\varphi$ be a model of $H$ in $G$ and let  $C':=\varphi(C)$.
Let $vw$ be an edge of $\varphi(H)$ having no neighbors in $C'$. Such an edge exists by the assumption that 
$H$ contains two adjacent vertices having no neighbors in $C$.

As $A$ is independent, one of $v$ and $w$ is not contained in $A$.
Without loss of generality, we assume $v\in V(H_P)\setminus A$ for some $P\in \mathcal{P}(UB_n)$.
Since $v$ has no neighbor in $C'$, 
$C'$ should be contained in $V(H_P)\cup A$.
By construction, $C'$ has at least $\abs{H}$ vertices in $H_P -A$.
Thus, if $\varphi(H)$ contained a vertex in $H_{P'} - A$ for some $P'\in \mathcal{P}(UB_n) \setminus \{P\}$, 
then it would have degree at least $\abs{H}$, which is not possible.
This implies that $V(\varphi(H))\subseteq V(H_P)\cup A$, and again by \autoref{lem:subdivisionlemma}, 
$\varphi(H)$ contains $P$.
The remaining part is the same as in \autoref{thm:cyclelength5}.
\end{proof}

\section{Negative results using the semi-grid}
\label{sec:segri}

In this section, we show that if $H$ contains a cycle and three vertices that have no neighbors in $C$, 
then $H$-subdivisions have no induced \ep{} property.
For example, we may consider a graph $F$ that is the disjoint union of $C_3$ and three isolated vertices.
Clearly, triangle-walls contain $F$ as an induced subgraph.
Also, in the construction based on hypergraphs in \autoref{sec:hypergraph}, 
we have an independent set $A$, and we can choose some $C_4$ by picking two vertices from each of $V(H_P)\setminus A$ and $V(H_{P'})\setminus A$ for some distinct $P, P'\in \mathcal{P}(UB_n)$, 
and then choose three vertices in $A$ having no neighbors in $C_4$ (for this, we choose $P, P'$ so that $A\setminus (V(P)\cup V(P'))$ has $3$ vertices).
Therefore, the constructions presented in the previous sections do not seem helpful to deal with this case. Let us introduce a new one.

For $n\in \mathbb{N}_{\ge 3}$, we define the \emph{semi-grid} $SG_n$ of order $n$ as follows:
\begin{enumerate}
\item $V(SG_n)=\{v_{i,j}:i, j\in \{1, \ldots, n\}, i\ge j\}$;
\item for every $i \in \intv{1}{n}$, we add the edges of the path $P_i$, defined as the concatenation of
  $
v_{i,1} \cdots v_{i,i} \quad \text{and} \quad v_{i,i} \cdots v_{n,i};
  $
\item additionally, we add an edge between two vertices $v$ and $w$ if there is no $i\in \{1, \ldots, n\}$ such that $v,w \in V(P_i)$.\label{e:compled}
\end{enumerate}

An example is depicted in \autoref{fig:sg6}.

\begin{figure}[h]
  \centering
  \begin{tikzpicture}
    \begin{scope}[every node/.style = black node]
    
      \foreach \i in {1,...,5}{
        \foreach \j in {1, ...,\i}{
          \draw (\j, -\i) node[label = 4:$v_{\i,\j}$] (N{\i,\j}) {};
        }
  }
  
      \foreach \i/\c in {2/purple, 3/green!50!black, 4/orange, 5/cyan}{
        \pgfmathtruncatemacro{\im}{int(\i-1)}
        \foreach \k in {1,...,\im}{
          \pgfmathtruncatemacro{\kp}{\k+1}
            \draw[very thick, color = \c] (N{\i,\k}) -- (N{\i,\kp});
          }
        }
        \foreach \i/\c in {1/blue, 2/purple, 3/green!50!black, 4/orange}{
          \foreach \k in {\i,...,4}{
            \pgfmathtruncatemacro{\kp}{\k+1}
            \draw[very thick, color = \c] (N{\k,\i}) -- (N{\kp,\i});
          }
        }
        \foreach \i/\c in {1/blue, 2/purple, 3/green!50!black, 4/orange, 5/cyan}{
          \draw (-0.25,-\i) node[normal, color = \c] {$P_{\i}$};
        }
      
    \end{scope}
  \end{tikzpicture}
  \caption[The graph $SG_5$.]{The graph $SG_5$. The $P_i$'s are drawn with different colors. The edges added in step \eqref{e:compled} are not depicted.}
  \label{fig:sg6}
\end{figure}
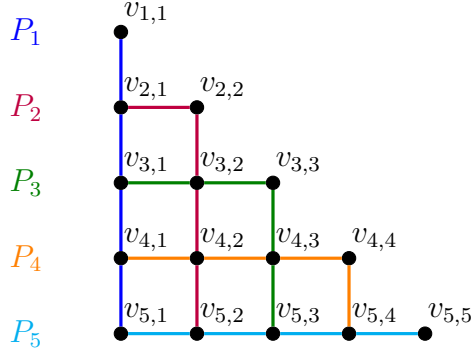

Observe that $SG_n$ satisfies the following, for every $i,j \in \intv{1}{n}$ with $i<j$:
\begin{itemize}
\item $P_i$ is an induced path;
\item for $i,j\in \{1,2, \ldots n\}$ with $i<j$, $V(P_i)\cap V(P_j) = \{v_{j,i}\}$;
\item every vertex of $SG_n$ belongs to at most two paths of $P_1, \ldots, P_n$;
\item $V(SG_n)\setminus V(P_i)$ has at least $n-2$ neighbors in $P_i$ (this is because a vertex in $V(SG_n)\setminus V(P_i)$ belongs to at most two paths of $P_1, \ldots, P_n$, which may contain non-neighbors of $v$ on their intersections with $P_i$). 
\end{itemize}

\begin{theorem}[item \eqref{badc:3v} of \autoref{thm:badcases}]\label{thm:cyclethreevertex}
Let $H$ be a graph that contains a cycle $C$ and three vertices having no neighbors in $C$.
Then the subdivisions of $H$ do not have the induced \ep{} property.
\end{theorem}

\begin{proof}
Let $m:=\abs{H}$.
We construct an infinite family $(G_n)_{n\in \mathbb{N}_{\ge m+2}}$ of graphs such that 
$\upnu_{H}(G_n) = 1$ while $\uptau_{H}(G_n) \geq \frac{n}{2}$ for every $n \in \N_{\geq m+2}$.

Let us fix an edge $uv$ of $C$. Let $n\ge m+2$ be an integer and let $H'$ be the graph obtained from $H$ by subdividing each edge into a path of length $n+1$. In particular, we subdivide the edge $uv$ into $uw_1w_2 \cdots w_n v$. The graph $G=G_n$ is constructed from $SG_n$ as follows:
\begin{itemize}
\item we take copies $H_1, H_2, \ldots, H_n$ of $H'$ and for each $j\in \{1, \ldots, n\}$, we identify the copy of $w_1w_2 \cdots w_n$ in $H_j$ with the path $P_j$;
\item then we add all possible edges between $V(H_j)\setminus V(P_j)$ and $V(G)\setminus V(H_j)$. 
\end{itemize}
Clearly, each $H_j$ is an induced subdivision of $H$. 
We show that every induced subdivision of $H$ in $G$ is $H_j$ for some $j\in \{1, \ldots, n\}$.

\begin{claim}\label{claim:fq2}
Every induced subdivision of $H$ in $G$ is contained in $H_j$ for some $j\in \{1, \ldots, n\}$ and contains the path $P_j$.
\end{claim}
\begin{proofclaim}
Let $\varphi$ be a model of $H$ in $G$, and let $F:=\varphi(H)$ and $C':=\varphi(C)$, and $Z=\{z_1, z_2, z_3\}$ be a set of three vertices in $F$ having no neighbors in $C'$. Such a set $Z$ exists 
because $H$ contains three vertices having no neighbors in $C$.

We consider two cases depending on whether $Z$ contains a vertex of $V(H_j)\setminus V(P_j)$ for some $j\in \{1, \ldots, n\}$, or not.

\medskip
\noindent

(\textit{First case:} $Z$ contains a vertex of $V(H_j)\setminus V(P_j)$ for some $j\in \{1,2, \ldots, n\}$.)

By construction, for all $a\in V(H_j)\setminus V(P_j)$ and $b\in V(G)\setminus V(H_j)$, $a$ is adjacent to $b$.
As $Z$ has no neighbors in $C'$, $C'$ is contained in $H_j$.
Furthermore, since $P_j$ is a path,
$C'$ contains a vertex of $V(H_j)\setminus V(P_j)$.
It implies that $Z$ is also contained in $H_j$.

First assume that $C'$ contains $P_j$.
In this case, every vertex in $V(G)\setminus V(H_j)$ has at least $n-2\ge m$ neighbors in $P_j$. As the maximum degree of $H$ is less than $m$, no vertex of $F$ belongs to $V(G)\setminus V(H_j)$. Hence $F$ is an induced subgraph of $H_j$. By assumption, $F$ contains $P_j$ so we are done.

In the remaining case, $C'$ does not contain $P_j$. As $V(C') \subseteq V(H_j)$ and $P_j$ is a path, we deduce that all vertices of $C'$ belong to $V(H_j)\setminus V(P_j)$.
By construction every vertex in $V(G)\setminus V(H_j)$ dominates $C'$.
Since $C'$ contains at least $m$ vertices while $H$ has maximum degree less than $m$, all other vertices of $F$ are contained in $H_j$.
If $F$ does not contain an edge in $P_j$, then $F$ contains no induced subdivision of $H$ by \autoref{lem:subdivisionlemma}.
Therefore, $F$ contains $P_j$, as required.

\medskip
\noindent

(\textit{Second case:} $Z \subseteq V(SG_n)$.)

First assume that $z_1, z_2, z_3$ are all contained in $P_i$ for some $i\in \{1, \ldots, n\}$. We can observe that every vertex of $V(G)\setminus V(P_i)$ has a neighbor in $\{z_1, z_2, z_3\}$. Thus, $C'$ is an induced subgraph of $H_i$. Since $P_i$ is a path and $z_1$ is already in $P_i$, $C'$ is contained in $V(H_i)\setminus V(P_i)$. 
Since $C'$ contains at least $m$ vertices and every vertex in $V(G)\setminus V(H_i)$ dominates $C'$,
all other vertices of $F$ are contained in $H_i$.
If $F$ does not contain an edge in $P_i$, then $F$ contains no induced subdivision of $H$ by \autoref{lem:subdivisionlemma}.
Therefore, $F$ contains $P_i$.

We can assume that two of $z_1, z_2, z_3$ are not contained in the same path $P_i$ of $SG_n$.
Without loss of generality, we assume that $z_1\in V(P_a)\cap V(P_{a'})$ and $z_2\in V(P_b)\cap V(P_{b'})$ such that 
\begin{itemize}
\item $\{a, a'\}\cap \{b, b'\}=\emptyset$, 
\item if $a=a'$, then $z_1:=v_{a,a}$,
\item if $b=b'$, then $z_2:=v_{b,b}$.
\end{itemize}
Therefore the four vertices $u_{a,b}\in V(P_a)\cap V(P_b)$, $u_{a, b'}\in V(P_a)\cap V(P_{b'})$, $u_{a', b}\in V(P_{a'})\cap V(P_b)$, and $u_{a', b'}\in V(P_{a'})\cap V(P_{b'})$ are the only vertices in $G$ that are not adjacent to both $z_1$ and $z_2$.
Thus, $V(C') \subseteq \{u_{a,b}, u_{a',b}, u_{a,b'}, u_{a',b'}\}$.

Observe that $z_3$ is contained in at least one and at most two paths of $P_a, P_{a'}, P_b, P_{b'}$.
If $z_3$ is contained in exactly one path of them, say $Q$, 
then $z_3$ is adjacent to the vertices of $\{u_{a,b}, u_{a',b}, u_{a,b'}, u_{a',b'}\}\setminus V(Q)$.
But every cycle in $\{u_{a,b}, u_{a',b}, u_{a,b'}, u_{a',b'}\}$ contains a vertex of $\{u_{a,b}, u_{a',b}, u_{a,b'}, u_{a',b'}\}\setminus V(Q)$, a contradiction.
We may assume that $z_3$ is contained in two paths of them, say $P_x$ and $P_y$.
Then $z_3=u_{x,y}$, and the other three vertices form a triangle $C'$.
But then $z_3$ is adjacent to the vertex in $\{u_{a,b}, u_{a',b}, u_{a,b'}, u_{a',b'}\}\setminus (V(P_x)\cup V(P_y))$, a contradiction.

This concludes the proof of the claim.
\end{proofclaim}

By \autoref{claim:fq2}, every induced subdivision  of $H$ in $G$ contains $P_j$ for some $j\in \{1, \ldots, n\}$.
Thus, two induced subdivisions of $H$ always intersect, and
hence $\upnu_{H}(G) = 1$. 
Let $S$ be a vertex subset of size less than $\frac{n}{2}$.
Clearly each vertex of $S$ hits at most two models in the graph.
Thus, $S$ hits less than $2(\frac{n}{2})=n$ models, and 
$G-S$ contains an induced subdivision of $H$.
This implies that $\uptau_{H}(G) \geq \frac{n}{2}$, as required.
\end{proof}

\section{Negative results for non-planar graphs}
\label{sec:plagra}

We prove here item \eqref{badc:nonpl} of \autoref{thm:badcases} (which is \autoref{l:notplanar}): subdivisions of a non-planar graph never have the induced \ep{} property.

We use the notion of \emph{Euler genus} of a graph $G.$
The {\em Euler genus} of a non-orientable surface ${\rm {\rm \Sigma}}$ is equal to the non-orientable genus
$\tilde{g}({\rm {\rm \Sigma}})$ (or the crosscap number). The {\em Euler genus}  of an orientable   surface
${\rm {\rm \Sigma}}$ is $2{g}({\rm {\rm \Sigma}})$, where ${g}({\rm {\rm \Sigma}})$ is  the orientable genus
of ${\rm {\rm \Sigma}}.$  We refer to the book of Mohar and Thomassen \cite{mohar2001graphs} for
more details  on graph embedding. The {\em Euler genus} $\gamma(G)$ of a graph $G$ is the minimum Euler genus of a surface where $G$ can be embedded.

We will adapt the proof of \cite[Lemma 5.2]{Raymond2017recent} (similar to that of \cite[Lemma~8.14]{RobertsonS86GMV} but dealing with subdivisions) to the setting of induced subdivisions.

\begin{lemma}[from the proof of {\cite[Lemma 5.2]{Raymond2017recent}}]\label{th:noepsub}
For every non-planar graph $H$, there is a family of graphs $(G_n)_{n \in \N}$ such that, for every $n\in \N$,
\begin{enumerate}[(i)]
\item $\gamma(G_n) = \gamma(H)$; and
\item $G_n - X$ contains a subdivision of $H$, for every $X\subseteq V(G_n)$ with $|X| < n$.\label{e:gpack}
\end{enumerate}
\end{lemma}

\begin{lemma}\label{lem:lb-nonplanar}\label{l:notplanar}
  For every non-planar graph $H$, the subdivisions of $H$ do not have the induced \ep{} property.
\end{lemma}

\begin{proof}
We construct a family of graphs $(\dot{G_n})_{n \in \N}$ such that $\upnu_H(\dot{G_n}) = 1$ and $\uptau_H(\dot{G_n})\geq n$, for every $n\in \N$.

  For every $n\in \N$, let $G_n$ be the graph of \autoref{th:noepsub} and let $\dot{G_n}$ be the graph obtained by subdividing once every edge of~$G_n$.
  Notice that for every subdivision $G'$ of a graph $G$ we have $\gamma(G) = \gamma(G')$ (an embedding of one in a given surface can straightforwardly be deduced from an embedding of the other).
  Hence $\gamma(\dot{G_n}) = \gamma(G_n) = \gamma(H)$.
  As in the proof of \cite[Lemma 5.2]{Raymond2017recent}), we observe that since $H$ is not planar the disjoint union of two subdivisions of $H$ has Euler genus larger than that of $H$ (see \cite{battle1962}). Therefore $\dot{G_n}$ does not contain two disjoint (induced) subdivisions of~$H$: $\upnu_H(\dot{G_n}) \leq 1$.
For every subgraph $J$ of $G_n$, the subgraph of $\dot{G_n}$ obtained by subdividing once every edge of $J$ is induced. Hence property \eqref{e:gpack} of $(G_n)_{n \in \N}$ implies that for every $X\subseteq V(\dot{G_n})$ with $|X| < n$ the graph $\dot{G_n} - X$ contains an induced subdivision of $H$. In other words $\uptau_H(\dot{G_n}) \geq n$.
\end{proof}

\section{A lower-bound on the bounding function}
\label{sec:lowbo}

We proved in the previous sections that subdivisions of graphs with certain properties do not have the induced \ep{} property.
In this section we give a negative result of a different type by
proving that the bounding function for the induced \ep{} property of
subdivisions cannot be too small.
Our proof is an adaptation to the induced setting of an observation
already known  for the non-induced \ep{} property.

We first need a few more definitions. For every graph $G$, the \emph{girth} of $G$, denoted by $\girth(G)$ is the minimum order of a cycle in~$G$.
The treewidth of $G$, denoted by $\tw(G)$ is a graph invariant that can be defined using tree-decompositions. We avoid the technical definition here and only state the two well-known properties of treewidth that we use:
\begin{itemize}
\item deleting a vertex or an edge in a graph decreases its treewidth by at most one;
\item for every planar graph $H$ of maximum degree 3, there is a constant $c\in \N$ such that every graph of treewidth at least $c$ contains a subdivision of $H$ (Grid Minor Theorem, \cite{RobertsonS86GMV}).
\end{itemize}
We refer the reader to \cite{diestel2010graph} for an introduction to treewidth.

\begin{theorem}\label{thm:omegaklogk}
  Let $H$ be a graph that has a cycle and no vertex of degree more than 3. There is no function $f(k) = o(k \log k)$ such that subdivisions of $H$ have the induced \ep{} property with bounding function~$f$.
\end{theorem}

\begin{proof}
  When $H$ is not planar, the result follows from \autoref{l:notplanar}. We therefore assume for now that $H$ is planar. By the Grid Minor Theorem, there is a constant $c$ such that every graph $G$ satisfying $\tw(G)\geq c$ contains a subdivision of $H$.
We will construct sequences $(G_n)_{n \in \N}$ (graphs) and $(k_n)_{n \in \N}$ (integers) such that $\upnu_H(G_n) = \mathcal{O}(k_n)$ while $\uptau_H(G_n) = \Omega(k_n \log k_n)$.
  
We start with an infinite family $(R_n)_{n \in \N}$ of 3-regular graphs (of increasing order), called Ramanujan graphs, whose existence is proved in \cite[Theorem 5.13]{Morgenstern1994exis}. These graphs have the following properties:
\begin{enumerate}
\item $\forall n \in \N$, $\girth(R_n) \geq  \frac{2}{3} \log |R_n|$ (\cite[Theorem 5.13]{Morgenstern1994exis});
\item $\exists n'\in \N,\ \exists \alpha\in \R_{>0},\ \forall n\in \N,\ n \geq n' \Rightarrow \tw(R_n) \geq \alpha|R_n|$ (see~\cite[Corollary~1]{Bezrukov2004155}).
\end{enumerate}

Recall that $n \mapsto |R_n|$ is increasing. We denote by $n''$ the minimum integer that is larger than $n'$ and such that $|R_n| \geq \frac{c}{\alpha}$ for every $n \geq n''$.
For every integer $n \geq n''$, we define $k_n$ as the maximum positive integer such that
\[
  \frac{1}{\alpha}\left (c + k_n \log k_n \right) \leq |R_{n}|
\]
Such a value exists by definition of $n''$. Notice that $n \mapsto k_n$ is non-decreasing and is not upper-bounded by a constant.
Observe that for every $n \geq n''$, $\tw(R_n) \geq c + k_n \log k_n$.

Let us define, for every integer $n \in \N$, $G_n$ as the graph obtained from $R_n$ by subdividing once every edge. We then have, for every $n\geq n''$:
\begin{align}
  |G_n| &= \frac{5}{2} \cdot |R_n|\nonumber\\
  & \geq \frac{5}{2\alpha}(c + k_n \log k_n) &\text{and} \label{eq:size}\\
  \girth(G_n) &= 2\cdot \girth(R_n)\nonumber\\
        &\geq \frac{4}{3}\cdot  \log\left ( \frac{2}{5} \cdot|G_n| \right ). \label{eq:girth}
\end{align}

\begin{claim}\label{claim:bigtau}
  For every integer $n \geq n''$, $\uptau_H(G_{n}) \geq k_n \log k_n$.
\end{claim}

\begin{proofclaim}
  Let $X \subseteq V(G_{n})$ be such that $|X|<k_n \log k_n$.
  We show that $G_{n} - X$ contains an induced subdivision of~$H$.
  Recall that each vertex of degree 2 of $G_{n}$ was obtained by subdividing an edge of $R_{n}$. We define $X^e$ as the set of edges of $R_{n}$ corresponding to the vertices of degree 2 in $X$ and set $X^v = X\cap V(R_{n})$.
  As the deletion of an edge or a vertex in a graph decreases its treewidth by at most one, the graph obtained from $R_{n}$ by deleting $X^v$ and $X^e$ has treewidth at least $c$. By definition of $c$, this graph contains a subdivision $S$ of $H$. Notice that the corresponding subdivision of $H$ in $G_{n}$ (i.e. that obtained by subdividing once every edge of $S$) is induced and does not contain any vertex of $X$.
  As this holds for every subset of $V(G_{n})$ with less than $k_n \log k_n$ vertices, we deduce that $\uptau_H(G_{n}) \geq k_n \log k_n$.  
\end{proofclaim}

\begin{claim}\label{claim:smallnu}
  There is a $n'''\in \N$, such that for every $n  \geq n'''$, $\upnu_H(G_{n}) < k_n$.
\end{claim}

\begin{proofclaim}
  Let $n \geq n''$ and let us assume that $G$ contains $k_n$ disjoint induced subdivisions of $H$. As $H$ has a cycle, the order of each of these subdivisions is at least the girth of~$G$.
  We deduce:
  \begin{align*}
    |G_n| &\geq k_n \cdot \girth(G_n)\\
    &\geq k_n \cdot \frac{4}{3}\cdot  \log\left ( \frac{2}{5} \cdot|G_n| \right ) & \text{by \eqref{eq:girth}}\phantom{.}\\
    &\geq k_n \cdot \frac{4}{3}\cdot  \log\left ( \frac{1}{\alpha}(c + k_n\log k_n)  \right )& \text{by \eqref{eq:size}}.
  \end{align*}

  On the other hand, from the maximality of $k_n$ we get:
  \[
    |G_n| < \frac{5}{2\alpha}(c + k_n \log k_n).
  \]
  Combining these bounds together we obtain:
  \begin{align}
    k_n \cdot \frac{4}{3}\cdot  \log\left ( \frac{1}{\alpha}(c + k_n\log k_n)  \right ) < \frac{5}{2\alpha}(c + k_n \log k_n).\label{eq:collision}
  \end{align}

  The left-hand side of \eqref{eq:collision} is a $\Omega(k_n \log (k_n \log k_n))$ while its right-hand side is a $\mathcal{O}(k_n \log k_n)$.
  As $n \mapsto k_n$ is not upper-bounded by a constant, there is a positive integer $n'''\geq n''$ such that for every $n\geq n'''$, \eqref{eq:collision} does not hold.
  Therefore, when $n \geq n'''$ we have $\upnu_H(G_{n}) < k_n$.
\end{proofclaim}

The sequences $(G_n)_{n \in \N}$ and $(k_n)_{n \in \N}$ have the property that $\uptau_H(G_n) = \Omega(k_n \log k_n)$ (\autoref{claim:bigtau}) while $\upnu_H(G_n) = \mathcal{O}(k_n)$ (\autoref{claim:smallnu}), as required.

\end{proof}

\section{Subdivisions of 1-pan or 2-pan have the induced \texorpdfstring{\ep{}}{Erdos-Posa} property}
\label{sec:panpan}

Recall that for $p\in \N_{\ge 1}$, the $p$-pan is the graph obtained by adding an edge between a triangle and an end vertex of a path on $p$ vertices.
We show in this section that for every $p\in \{1,2\}$, the subdivisions of the $p$-pan have the induced \ep{} property with bounding function $\mathcal{O}(k^2 \log k)$ (\autoref{t:1panep} and \autoref{t:2panep}). 
For $p\ge 3$, the $p$-pan has an edge that has no neighbor in its triangle, 
and by \autoref{thm:twoadjacent}, the subdivisions of the $p$-pan do not have the induced \ep{} property.
This proves the dichotomy \eqref{enum:dich-pan} of \autoref{thm:dichotomy}.
Furthermore, the proofs of our positive results yield polynomial-time algorithms for finding a packing of induced subdivisions of the $p$-pan or a hitting set of size $\mathcal{O}(k^2 \log k)$.

The proofs for $1$-pan and $2$-pan start similarly.
Let $p\in \{1,2\}$ and let $G$ and $H$ be graphs such that:
\begin{itemize}
\item $H$ is an induced subdivision of the $p$-pan in $G$ with minimum number of vertices, 
\item $G- V(H)$ has no induced subdivision of the $p$-pan.
\end{itemize}
With these assumptions, we will show that for every $k\in \N$, $G$ contains either $k$ pairwise vertex-disjoint induced subdivisions of the $p$-pan, 
or a vertex set $S$ of size $\mathcal{O}(k \log k)$ hitting all induced subdivisions of the $p$-pan.
By applying inductively this result we can conclude that a graph contains either $k$ pairwise vertex-disjoint induced subdivisions of the $p$-pan, 
or a vertex set of size $\mathcal{O}(k^2 \log k)$ hitting all induced subdivisions of the $p$-pan.

The following algorithm to find an induced subdivision of the $p$-pan will be necessary.

\begin{lemma}\label{lem:algorithm}
Let $p$ be a positive integer. Given a graph $G$, one can find an induced subdivision of the $p$-pan with minimum number of vertices, if one exists, in time $\mathcal{O}(\abs{G}^{p+5})$.
\end{lemma}
\begin{proof}
  We first describe how, given vertices $v_1, v_2, \ldots, v_p$ and $w_1, w_2, w_3$, one can in $\mathcal{O}(|G|^2)$ steps
  determine whether there exists 
  an induced subdivision of the $p$-pan $H$ such that 
  \begin{itemize}
  \item $\{v_1, v_2 \ldots, v_p, w_1, w_2, w_3\}\subseteq V(H)$, 
  \item $v_1\dots v_pw_2$ is an induced path of $G$, 
  \item both $w_1$ and $w_3$ are adjacent to $w_2$,
  \item $H-\{v_1, v_2, \ldots, v_p, w_2\}$ is an induced path of $G$ from $w_1$ to $w_3$,
  \end{itemize}
  and output such an induced subdivision with minimum order if one exists.
  Since it is an induced subdivision of the $p$-pan,
  there should not exist an edge between $\{v_1, v_2, \ldots, v_p\}$ and the induced path $H-\{v_1, v_2, \ldots, v_p, w_2\}$, 
  and furthermore, there should not exist an edge between $w_2$ and the internal vertices of the path.

We may check the first three conditions in time $\mathcal{O}(p^2)$.
If one of them is not satisfied, then we answer negatively.

We then compute a shortest path $P$ from $w_1$ to $w_3$ in 
\[G-\{v_1, v_2, \ldots, v_p, w_2\}-(N_G(\{v_1, v_2, \ldots, v_p, w_2\})\setminus \{w_1, w_3\}).\] 
It can be done in time $\abs{G}^2$, for instance using Dijkstra's algorithm. If such a path does not exist, then we answer negatively.
Otherwise, $G[\{v_1, v_2, \ldots, v_p, w_2\}\cup V(P)]$ is the desired induced subdivision of the $p$-pan.

We apply the above procedure for every choice of $p+3$ vertices of $G$.
In the end, we output an induced subdivision of the $p$-pan of minimum order among those returned, if any. Overall this takes $\mathcal{O}(\abs{G}^{p+5})$ steps.
If $S$ is an induced subdivision of the $p$-pan in $G$, then one choice of $v_1v_2 \cdots v_pw_2$ corresponds to the path of length $p$ of $S$ and $w_1$ and $w_3$ to the two neighbors of $w_2$ on the cycle. A shortest path from $w_1$ to $w_3$ yields an induced subdivision of the $p$-pan of order at most $|S|$. This guarantees that the algorithm outputs an induced subdivision of the $p$-pan with minimum order.
\end{proof}

\subsection{On induced subdivisions of the 1-pan}
The aim of this section is to prove the following theorem.
\begin{theorem}\label{t:1panep}
There is a polynomial-time algorithm 
which, given a graph $G$ and a positive integer $k$, finds either $k$ vertex-disjoint induced subdivisions of the $1$-pan in $G$ or a vertex set of size $\mathcal{O}(k^2\log k)$ hitting every induced subdivision of the $1$-pan in $G$.
\end{theorem}

 Let us refer to an induced subdivision of the $1$-pan as a pair $(v, C)$, where $v$ is the vertex of degree one and $C$ is the cycle. We start with structural lemmas.

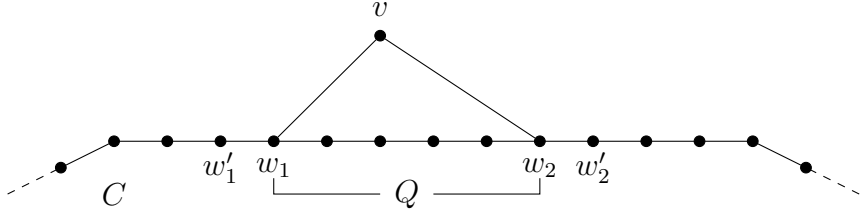
\begin{figure}
  \centering
  \begin{tikzpicture}[scale=0.7]
  \tikzstyle{w}=[circle,draw,fill=black,inner sep=0pt,minimum width=4pt]

   \draw (-2,0)--(10,0);
	\draw(-2, 0)--(-3,-0.5);
	\draw(10, 0)--(11,-0.5);
      \draw (-3,-.5) node [w] {};
       \draw (11,-.5) node [w] {};
 	\draw[dashed](12, -1)--(11,-0.5);
	\draw[dashed](-4,-1)--(-3,-0.5);
 
 \foreach \y in {-2,...,10}{
      \draw (\y,0) node [w] (a\y) {};
     }
       \draw (3,2) node [w] (x) {};
	\draw  (x)--(a1);
	\draw  (x)--(a6);
      \node at (-2, -1) {$C$};
       \node at (3, 2.5) {$v$};
       \node at (1, -.5) {$w_1$};
       \node at (6, -.5) {$w_2$};
       \node at (0, -.5) {$w_1'$};
       \node at (7, -.5) {$w_2'$};
       \node at (3.5, -1) {$Q$};
       \draw(1, -.7)--(1,-1)--(3,-1);
      \draw(6, -.7)--(6,-1)--(4,-1);

   \end{tikzpicture}     \caption{Two neighbors $w_1$ and $w_2$ of $v$ with minimum $\dist_C(w_1, w_2)$ in \autoref{l:oneneigh}.}\label{fig:distance}
\end{figure}

\begin{lemma}\label{l:oneneigh}
  Let $G$ be a graph and let $H=(u,C)$ be an induced subdivision of the $1$-pan in $G$ with minimum number of vertices. 
  If $\abs{C}\geq 5$, then every vertex of $V(G) \setminus V(H)$ has at most one neighbor in~$C$.
\end{lemma}
\begin{proof}
 Suppose for contradiction that there exists a vertex $v\in V(G)\setminus V(H)$ having at least two neighbors in $C$. 
 We choose two neighbors $w_1, w_2$ of $v$ with minimum $\dist_C(w_1, w_2)$.
 Let $Q$ be a shortest path from $w_1$ to $w_2$ in $C$.
 See \autoref{fig:distance} for an illustration.
Then $G[V(Q)\cup \{v\}]$ is an induced cycle, 
and it is strictly shorter than $C$, as $\abs{C}\ge 5$ and 
\[\abs{Q}+1\le \frac{\abs{C}}{2}+2< \abs{C}.\]

If $V(C)=V(Q)$, then $w_1$ is adjacent to $w_2$, and by the choice of $Q$, 
$Q$ also has length $1$. This is not possible. So we know that $C-V(Q)$ has at least one vertex.

For each $i\in \{1,2\}$, let $w_i'$ be the neighbor of $w_i$ in $C$, which is not on the path $Q$.
If $w_1'=w_2'$, then $Q$ has length at most $2$, and $|C|\le 4$, a contradiction.
So, we may assume that $w_1'\neq w_2'$.

Since $H$ is an induced subdivision with minimum number of vertices, for each $i\in \{1, 2\}$, 
$G[V(Q)\cup \{v,w_i'\}]$ is not an induced subdivision of the 1-pan.
It implies that $v$ is adjacent to both $w_1'$ and $w_2'$.
Note that $w_1'$ is not adjacent to $w_2'$; otherwise, $\dist_C(w_1, w_2)\le \dist_C(w_1', w_2')=1$ and
we would have $\abs{C}\le 4$, a contradiction. 
Therefore, $G[\{v, w_1, w_1', w_2'\}]$ is isomorphic to the $1$-pan, contradicting the assumption that $H$ is an induced subdivision with minimum number of vertices. 

We conclude that every vertex in $V(G)\setminus V(H)$ has at most one neighbor in $C$.
\end{proof}

Suppose $(u,C)$ is an induced subdivision of the $1$-pan in a graph $G$.
In the next lemmas, we explain how to extract many induced subdivisions of the 1-pan
from a set of $V(C)$-paths of $G-E(C)$.
 For an induced cycle $U$ of $G$ and $q\in V(G) \setminus V(U)$, we call $(q,U)$ a \emph{good pair} if $G[\{q\}\cup V(U)]$ contains an induced subdivision of the 1-pan.

\begin{lemma}\label{lem:induced1pan}
Let $C=v_1v_2 \cdots v_mv_1$ be an induced cycle of length at least $4$ in a graph $G$, and let $v\in V(G)\setminus V(C)$ such that 
$v$ is adjacent to $v_3$, and non-adjacent to $v_1, v_2, v_4$.
Then $G[V(C)\cup \{v\}]$ contains an induced subdivision of the $1$-pan.
\end{lemma}
\begin{proof}
If $v$ has no neighbors in $V(C)\setminus \{v_1, v_2, v_3, v_4\}$, then it is clear. 
We may assume $v$ has a neighbor in $V(C)\setminus \{v_1, v_2, v_3, v_4\}$.
We choose a neighbor $v_i$ of $v$ in $V(C)\setminus \{v_1, v_2, v_3, v_4\}$ with minimum $i$.
Then $G[\{v, v_2, v_3, v_4, \ldots, v_i\}]$ is an induced subdivision of the $1$-pan in $G$. 
\end{proof}

\begin{lemma}\label{l:panpack}
  Let $k\ge 2$ be an integer, let $G$ be a graph, and $(v,C)$ be an induced subdivision of the $1$-pan in $G$ with minimum number of vertices.
  Given a set  $\mathcal{P}$ of vertex-disjoint $V(C)$-paths in  $G-E(C)$ with $\abs{\mathcal{P}}\ge 108 k \log k$, one can find  $k$ vertex-disjoint induced subdivisions of the 1-pan in time $\mathcal{O}(\abs{G}^3)$.
\end{lemma}

\begin{proof}
 The existence of such a set $\mathcal{P}$ guarantees that $C$ has length at least $5$.

For each $P\in \mathcal{P}$, let $\last(P)$ be the set of end vertices of $P$.
We construct a subset $\mathcal{P}'$ of $\mathcal{P}$ with the following property:
\begin{itemize}
\item $\forall P_1, P_2\in \mathcal{P}',\ \dist_C(\last(P_1), \last(P_2))\ge 3$,
\item $\abs{\mathcal{P}'}\ge 12k\log k$.
\end{itemize}
This can be done by repeatedly choosing a path $P$ in $\mathcal{P}$ and discarding from $\mathcal{P}$ all paths that have an endpoint at distance at most 4 from one of $\last(P)$ in $C$. For each path added to $\mathcal{P}'$, at most 8 are discarded, hence $\abs{\mathcal{P}'}\ge \abs{\mathcal{P}}/9\ge 12k\log k$.

  We consider the subgraph $H$ on the vertex set $V(C)\cup (\bigcup_{P\in \mathcal{P'}}V(P))$ and edge set $E(C)\cup (\bigcup_{P\in \mathcal{P'}}E(P))$.
  This graph has maximum degree $3$ and has at least $24k \log k$ vertices of degree $3$, as each path of $\mathcal{P'}$ contributes for two vertices of degree $3$.

  According to \autoref{t:simonovits}, one can in time $\abs{G}^3$ construct a set $\mathcal{Q}'$ of $k$ vertex-disjoint cycles in $H$. 
  Observe that $C$ intersects all other cycles, and thus, $C$ is not contained in $\mathcal{Q}'$.
  Note that for each cycle $U$ of $\mathcal{Q}'$, $G[V(U)]$ contains an induced cycle that has at least one edge of~$C$.
  This is always possible: in the cycle $U$, every chord $e$ divides the cycle into two paths, one of which, together with $e$, is again a cycle containing an edge of $C$ and less chords.
  Let $\mathcal{Q}$ be a collection of $k$ resulting induced cycles of $G$ obtained from~$\mathcal{Q}'$.

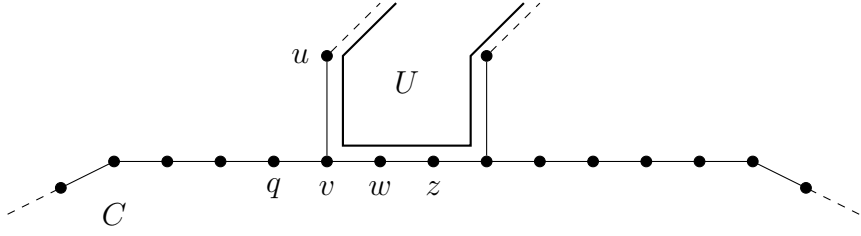
\begin{figure}
  \centering
  \begin{tikzpicture}[scale=0.7]
  \tikzstyle{w}=[circle,draw,fill=black,inner sep=0pt,minimum width=4pt]

   \draw (-2,0)--(10,0);
	\draw(-2, 0)--(-3,-0.5);
	\draw(10, 0)--(11,-0.5);
      \draw (-3,-.5) node [w] {};
       \draw (11,-.5) node [w] {};
 	\draw[dashed](12, -1)--(11,-0.5);
	\draw[dashed](-4,-1)--(-3,-0.5);
 
 \foreach \y in {-2,...,10}{
      \draw (\y,0) node [w] (a\y) {};
     }
       \draw (2,2) node [w] (x1) {};
       \draw (5,2) node [w] (x2) {};
	\draw  (x1)--(a2);
	\draw  (x2)--(a5);
	\draw[dashed] (x1)--(3,3);
	\draw[dashed] (x2)--(6,3);
	
	\draw[thick] (3.3,3)--(2.3, 2)--(2.3,0.3)--(4.7,0.3)--(4.7, 2)--(5.7,3);
	
      \node at (-2, -1) {$C$};
      \node at (3.5, 1.5) {$U$};
      \node at (2, -.5) {$v$};
      \node at (3, -.5) {$w$};
      \node at (4, -.5) {$z$};
      \node at (1.5, 2) {$u$};
          
     \node at (1, -.5) {$q$};

   \end{tikzpicture}     \caption{Four consecutive vertices $u,v,w,z$ of $U$ and the neighbor $q$ of $v$ in $C$ other than $w$, described in \autoref{l:panpack}.
   Since $U$ is an induced cycle, such a neighbor $q$ exists.}\label{fig:cycle}
\end{figure}

   For each cycle $U\in \mathcal{Q}$, there are four consecutive vertices $u, v, w, z$ of $U$ such that $vw\in E(C)\cap E(U)$ and $u\notin V(C)$.
   See \autoref{fig:cycle} for an illustration.
   Note that the neighbor of $v$ in $C$ other than $w$ is not contained in $U$, because $U$ has no vertex of degree $3$ in $G$.
 Let $q$ be the neighbor of $v$ in $C$ other than $w$.
 Observe that $q$ is adjacent to $v$ but has no neighbors in $\{u,w,z\}$ by \autoref{l:oneneigh}.
 Indeed, if $z$ is in $C$, then $q$ is not adjacent to $z$ because $\abs{C}\ge 5$,
 and if $z$ is not in $C$, then $q$ is not adjacent to $z$ because $q$ has a neighbor $w$ in $C$.
 Therefore, $G[\{q\}\cup V(U)]$ contains an induced subdivision of the $1$-pan by \autoref{lem:induced1pan}, that is, $(q, U)$ is a good pair. 

	Following the above procedure, for each $U_i\in \mathcal{Q}$, 
	we choose a vertex $q_i$ as $q$.
	Observe that for distinct cycles $U_i, U_j\in \mathcal{Q}$, 
	$q_i\neq q_j$ because otherwise 
	\[\dist_C(V(U_i)\cap V(C), V(U_j)\cap V(C))\le 2,\] 
	and it implies that for some $P, P'\in \mathcal{P}'$,
	\[\dist_C(\last(P), \last(P'))\le 2.\] 
	It contradicts our choice of $\mathcal{P}'$.
	Therefore, $q_1, q_2, \ldots, q_k$ are distinct vertices, and 
    \[(q_1, U_1), (q_2, U_2), \ldots, (q_k, U_k)\] are pairwise disjoint $k$ good pairs.
	Using \autoref{lem:induced1pan}, one can output $k$ vertex-disjoint induced subdivisions of the $1$-pan in linear time.
\end{proof}	

We are now ready to prove the main result of this section.
\begin{proof}[Proof of \autoref{t:1panep}.]
If $k=1$, then there is nothing to show. We assume that $k\ge 2$.

We construct sequences of graphs $G_1,\ldots, G_{\ell+1}$ and $F_1,\ldots , F_{\ell}$ 
with maximum $\ell$
such that 
\begin{itemize}
\item $G_1=G$, 
\item for every $i\in \{1, \ldots, \ell\}$, $F_i$ is an induced subdivision of the $1$-pan in $G_i$ with minimum number of vertices,
\item for every $i\in \{1, \ldots, \ell\}$, $G_{i+1}=G_i-V(F_i)$.
\end{itemize}
Such a sequence can be constructed in polynomial time repeatedly applying \autoref{lem:algorithm}.
If $\ell \geq k$, then we have found a  packing of $k$ induced subdivisions of the $1$-pan. 
Hence, we may assume that $\ell \leq k-1$. 

Let $\mu_k:=216k\log k +12k-11$. The rest of the proof relies on the following claim.

\begin{claim}
Let $j \in \intv{1}{\ell+1}$. One can find in polynomial time either $k$ vertex-disjoint induced subdivisions of the $1$-pan, or 
a vertex set $X_{j}$ of $G_{j}$ of size at most $(\ell+1-j)\mu_k$ such that $G_j-X_j$ has no induced subdivision of the $1$-pan. 
\end{claim}

\begin{proofclaim}
  We prove the claim by induction for $j=\ell +1$ down to $j=1$.
The claim trivially holds for $j=\ell+1$ with $X_{\ell+1}=\emptyset$ because $G_{\ell+1}$ has no induced subdivision of the $1$-pan.
Let us  assume that for some $j\leq \ell$, we obtained a required vertex set $X_{j+1}$ of $G_{j+1}$ of size at most $(\ell-j)\mu_k$.
Then in $G_{j}-X_{j+1}$, $F_{j}$ is an induced subdivision of the $1$-pan with minimum number of vertices.
If $F_{j}$ has at most $5$ vertices, then we set $X_{j}:=X_{j+1}\cup V(F_{j})$. Clearly, $\abs{X_j}\leq (\ell-j+1)\mu_k$.
We may thus assume that $F_j$ has at least $6$ vertices. Let $F_j:=(u,C)$. Observe that $\abs{C}\ge 5$.

We first apply Gallai's $A$-path Theorem (\autoref{t:gallaiapath}) with $A=V(C)$ for finding at least $108k \log k$ pairwise vertex-disjoint $V(C)$-paths in $(G_j-X_{j+1})-E(C)$.
Assume it outputs such $V(C)$-paths.
Then, by applying \autoref{l:panpack} to $G_{j}-X_{j+1}$ and $V(C)$,
one can find in polynomial time $k$ vertex-disjoint induced subdivisions of the $1$-pan.
Thus, we may assume that \autoref{t:gallaiapath} outputs a vertex set $S$ of size at most $216k\log k$ hitting all $V(C)$-paths in $(G_j-X_{j+1})-E(C)$.
 
Now, we consider the graph $G_j':=G_j-(X_{j+1}\cup S\cup \{u\})$. 
Suppose $G_j'$ contains an induced subdivision $Q=(d,D)$ of the $1$-pan.
Then $G_j'[V(C)\cap V(Q)]$ is connected; otherwise, $G_j'$ contains a $V(C)$-path in $(G_j-X_{j+1})-E(C)$, contradicting with that $S$ hits all such $V(C)$-paths.
Furthermore, $G_j'[V(C)\cap V(Q)]$ contains no edge of $D$; otherwise, we also have a $V(C)$-path in $(G_j-X_{j+1})-E(C)$.
Thus, we have that $\abs{V(C)\cap V(Q)}\le 2$.

To hit such remaining induced subdivisions of the $1$-pan,
we recursively construct sets $U\subseteq V(C)\setminus S$ and $\mathcal{J}$ as follows. At the beginning, set $U:=\emptyset$ and $\mathcal{J}:=\emptyset$.
For every set of four consecutive vertices $v_1, v_2, v_3, v_4$ of $C$ with $v_2, v_3\notin U$ and $v_2, v_3\notin S$, 
we test whether $G_j'[(V(G_j')\setminus V(C))\cup \{v_2,v_3\}]$ contains an induced subdivision $H$ of the $1$-pan, 
and if so, add vertices $v_1, v_2, v_3, v_4$ to $U$, and add $H$ to $\mathcal{J}$ and increase the counter by 1. If the counter reaches $k$, then we stop. 

Observe that graphs in $\mathcal{J}$ are pairwise vertex-disjoint; otherwise, we can find a $V(C)$-path in $(G_j-X_{j+1})-E(C)$, a contradiction.
Thus, if the counter reaches $k$, we can output $k$ pairwise vertex-disjoint induced subdivisions of the $1$-pan in polynomial time.

Assume that the counter does not reach $k$. Then
we have $\abs{U}\leq 4(k-1)$. We take the $1$-neighborhood $U'$ of $U$ in $C$, and observe that $\abs{U'}\le 12(k-1)$. 
We claim that $G_j'-U'$ has no induced subdivision of the $1$-pan. 
Suppose for contradiction that there is an induced subdivision $F$ of the $1$-pan in $G_j'-U_j$.
The intersection of $F$ on $C$ is a set of at most two consecutive vertices, say $T$.
In case when $\abs{T}=1$, neighbors of $T$ in $C$ are not contained in $U$, as $U'$ is the $1$-neighborhood of $U$ in $C$.
Thus, we can increase the counter by adding this induced subdivision to $\mathcal{J}$, a contradiction.
Therefore, $G_j'-U'$ has no induced subdivision of the $1$-pan, 
and 
\[X_j:=X_{j+1}\cup (S\cup \{u\}\cup U')\] 
satisfies the claim.
\end{proofclaim}

The result follows from the claim with $j=1$.
\end{proof}

\subsection{On induced subdivisions of the 2-pan}
This section is devoted to the proof of the following result.
\begin{theorem}\label{t:2panep}
There is a polynomial-time algorithm 
which, given a graph $G$ and a positive integer $k$,   finds either $k$ vertex-disjoint induced subdivisions of the $2$-pan in $G$ or a vertex set of size $\mathcal{O}(k^2\log k)$ hitting every induced subdivision of the $2$-pan in $G$.
\end{theorem}
We refer to an induced subdivision of the $2$-pan as a tuple $(v_1, v_2, C)$, where $v_1$, $v_2$ are respectively the vertices of degree one and two that are not contained in the cycle, and $C$ is the cycle.

As in the $1$-pan case, we first obtain a structural property.

\begin{lemma}\label{lem:dominating}
  Let $G$ be a graph and let $H=(v_1, v_2, C)$ be an induced subdivision of the $2$-pan in $G$ with minimum number of vertices such that $\abs{C}\ge 11$. 
  For every vertex $v$ in $V(G)\setminus V(H)$, either $v$ has at most one neighbor in $C$ or it dominates $C$. 
\end{lemma}

\begin{figure}
  \centering
  \begin{tikzpicture}[scale=0.7]
  \tikzstyle{w}=[circle,draw,fill=black,inner sep=0pt,minimum width=4pt]

   \draw (-2,0)--(10,0);
	\draw(-2, 0)--(-3,-0.5);
	\draw(10, 0)--(11,-0.5);
      \draw (-3,-.5) node [w] {};
       \draw (11,-.5) node [w] {};
 	\draw[dashed](12, -1)--(11,-0.5);
	\draw[dashed](-4,-1)--(-3,-0.5);
 
 \foreach \y in {-2,...,10}{
      \draw (\y,0) node [w] (a\y) {};
     }
       \draw (3,2) node [w] (x) {};
	\draw  (x)--(a1);
	\draw  (x)--(a6);
      \node at (-2, -1) {$C$};
       \node at (3, 2.5) {$v$};
       \node at (1, -.5) {$w_1$};
       \node at (6, -.5) {$w_2$};
       \node at (0, -.5) {$x_1$};
       \node at (-1, -.5) {$y_1$};
       \node at (7, -.5) {$x_2$};
       \node at (8, -.5) {$y_2$};
       \node at (3.5, -1) {$Q$};
       \draw(1, -.7)--(1,-1)--(3,-1);
      \draw(6, -.7)--(6,-1)--(4,-1);

   \end{tikzpicture}     \caption{Two neighbors $w_1$ and $w_2$ of $v$ with minimum $\dist_C(w_1, w_2)$ in \autoref{lem:dominating}.}\label{fig:distance2}
\end{figure}
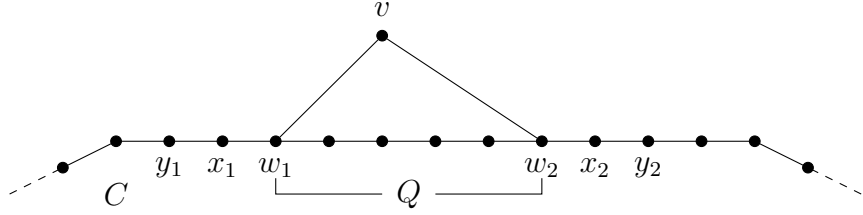

\begin{proof}
 Suppose there exists a vertex $v\in V(G)\setminus V(H)$ such that it has at least two neighbors in $C$ and has a non-neighbor~$z$.
We choose two neighbors $w_1, w_2$ of $v$ and a non-neighbor $z$ such that
	\begin{enumerate}
	\item $\dist_C(w_1, w_2)$ is minimum, 
	\item subject to (1), $\dist_C(w_1, z)+\dist_C(w_2, z)$ is minimum. 
	\end{enumerate}
	Let $Q$ be a shortest path of $C$ from $w_1$ to $w_2$.
Then $G[V(Q)\cup \{v\}]$ is an induced cycle, and it is strictly shorter than $C$, because $\abs{Q}+1\le \abs{C}/2+2< \abs{C}$.

For each $i\in \{1,2\}$, let $x_i$ be the neighbor of $w_i$ in $C$, which is not on the path $Q$, 
and let $y_i$ be the neighbor of $x_i$ in $C$ other than $w_i$. See \autoref{fig:distance2} for an illustration.
Note that $\{x_1, y_1\}\cap \{x_2, y_2\}=\emptyset$ and furthermore, 
there are no edges between $\{x_1, y_1\}$ and $\{x_2, y_2\}$; otherwise, $C$ would have length at most $10$, a contradiction.

Observe that for each $i\in \{1, 2\}$, $v$ has a neighbor in $\{x_i, y_i\}$; otherwise, $G[V(Q)\cup \{v,x_i,y_i\}]$ is an induced subdivision of the $2$-pan, which has smaller number of vertices than $H$. 

We claim that $v$ is complete to $\{x_1, y_1, x_2, y_2\}$.
Suppose not. Without loss of generality, we may assume that $v$ has a non-neighbor in $\{x_1, y_1\}$. 
Since $v$ has a neighbor in $\{x_2, y_2\}$, 
there is an induced cycle $C'$ in $G[\{v, w_2, x_2, y_2\}]$ that contains~$v$.
Depending which of $x_1,y_1$ is not adjacent to $v$, one of $(x_1, y_1, C')$ and $(y_1, x_1, C')$ is an induced subdivision of the 2-pan with less vertices than $H$.
This is a contradiction, and 
consequently, $v$ is complete to $\{x_1, y_1, x_2, y_2\}$.

By the choice of $w_1,w_2$, this also implies that $w_1$ and $w_2$ are adjacent.
But then $z\in V(C)\setminus \{x_1, y_1, w_1, w_2, x_2, y_2\}$, and 
$z$ is closer on $C$ to one of $\{x_1, y_1\}$ and $\{x_2, y_2\}$ than to $\{w_1, w_2\}$.
As this contradicts the choice of $w_1, w_2, z$, we conclude that $v$ dominates $C$.
\end{proof}

The following is a 2-pan counterpart of \autoref{lem:induced1pan}.

\begin{lemma}\label{lem:induced2pan}
Let $C=v_1v_2 \cdots v_mv_1$ be an induced cycle of length at least $4$ in a graph $G$ and $w_1, w_2, w_3, w_4\in V(G)\setminus V(C)$ such that 
\begin{enumerate}
\item $v_2w_1w_2w_3w_4$ is an induced path,
\item $v_2w_1$ is the only edge between $\{v_1, v_2, v_3, v_4\}$ and $\{w_1, w_2, w_3, w_4\}$, and 
\item each vertex in $C$ has at most one neighbor in $\{w_1, w_2, w_3, w_4\}$.
\end{enumerate}
Then $G[V(C)\cup \{w_1, w_2, w_3, w_4\}]$ contains an induced subdivision of the $2$-pan.
\end{lemma}
\begin{proof}
If $v_2w_1$ is the only edge between $\{w_1, w_2\}$ and $C$, then
$(w_1, w_2, C)$ is an induced subdivision of the $2$-pan.
Thus, we may assume that $C$ has length at least $5$, and $w_1$ or $w_2$ has a neighbor in $C - \{v_1, v_2, v_3, v_4\}$.
We choose $i\in\intv{5}{m}$ such that:
\begin{itemize}
\item $v_i$ has a neighbor in $\{w_1, w_2\}$; and 
\item no internal vertex of the path $P = v_iv_{i+1} \dots v_mv_1v_2$ has a neighbor in $\{w_1, w_2\}$.
\end{itemize}
We may assume that $i=5$; otherwise, the induced cycle in $G[V(P)\cup \{w_1, w_2\}]$ that contains $v_i$ together with $v_3, v_4$ forms an induced subdivision of the $2$-pan.
On the other hand, we can observe that $P$ contains at most $2$ internal vertices; otherwise the induced cycle in $G[\{w_1, w_2, v_2, v_3, v_4, v_5\}]$ that contains $v_i$ together with two internal vertices in $P$ forms an induced subdivision of the $2$-pan.

We distinguish cases depending the number of internal vertices of $P$.

\medskip
\noindent
\textit{First case:} $P$ has a unique internal vertex, i.e.\ $m=5$. As $v_5$ has a neighbor in $\{w_1, w_2\}$, 
it has none in $\{w_3, w_4\}$ by our assumption. Therefore, $G[V(P)\cup \{w_1, w_2, w_3, w_4\}]$ contains an induced subdivision of the $2$-pan.

\medskip
\noindent
\textit{Second case:} $P$ has exactly $2$ internal vertices, i.e.\ $m=6$. 
Note that $v_6$ has no neighbors in $\{w_1, w_2\}$.
If $v_6$ has a neighbor in $\{w_3, w_4\}$, then the induced cycle in $G[\{v_6, v_1, v_2, w_1, w_2, w_3, w_4\}]$ that contains $v_2$ together with $v_3, v_4$ induces a subdivision of the $2$-pan.
Hence we may assume that $v_6$ has no neighbors in $\{w_1, w_2, w_3, w_4\}$.
Since $v_5$ has no neighbor in $\{w_3, w_4\}$ (as above), $G[V(P)\cup \{w_1, w_2, w_3, w_4\}]$ contains an induced subdivision of the $2$-pan.
\medskip

This concludes the lemma.
\end{proof}

\begin{lemma}\label{lem:2panpack}
  Let $k$ be a positive integer, $G$ be a graph, $(v_1, v_2 ,C)$ be an induced subdivision of the $2$-pan in $G$ with minimum number of vertices, and $D$ be the set of all vertices dominating $C$. 
  Given a set $\mathcal{P}$ of vertex-disjoint $V(C)$-paths in  $G-E(C)-D$ with $\abs{\mathcal{P}}\ge 396k\log k$, one can find in polynomial time $k$ vertex-disjoint induced subdivisions of the $2$-pan.
\end{lemma}
\begin{proof}
For each $P\in \mathcal{P}$, let $\last(P)$ be the set of end vertices of $P$.
We construct a subset $\mathcal{P}'$ of $\mathcal{P}$ with the following property:
\[\forall P_1, P_2\in \mathcal{P}',\ \dist_C(\last(P_1), \last(P_2))\ge 9.\]
This can be done by repeatedly choosing a path $P$ in $\mathcal{P}$ and discarding from $\mathcal{P}$ all paths $Q$ that have an endpoint at distance at most 8 from one of $P$. 
For each path added to $\mathcal{P}'$, at most 32 are discarded, hence $\abs{\mathcal{P}'}\ge \abs{\mathcal{P}}/33\ge 12k\log k$.

  We consider the subgraph $H$ on the vertex set $V(C)\cup (\bigcup_{P\in \mathcal{P}'}V(P))$ and edge set $E(C)\cup (\bigcup_{P\in \mathcal{P}'}E(P))$.
  This graph has maximum degree $3$ and has at least $24k \log k$ vertices of degree $3$, as each path of $\mathcal{P}'$ contributes for two vertices of degree $3$. 	
  According to \autoref{t:simonovits}, one can in polynomial time construct a set $\mathcal{Q}$ of $k$ vertex-disjoint cycles of $C$. 
  Observe that $C$ intersects all other cycles, and thus, $C$ is not contained in $\mathcal{Q}$.
  
As in the proof of \autoref{l:panpack} we note that for each cycle $U$ of $\mathcal{Q}$, $G[V(U)]$ has an induced cycle containing at least one edge of $C$.
  Let $\mathcal{Q}'$ be a collection of $k$ resulting induced $V(C)$-cycles obtained from $\mathcal{Q}$.
  
  For each cycle $U$ in $\mathcal{Q}'$, 
  we want to find $4$ consecutive vertices of $C$ that are not contained in $\bigcup_{F\in \mathcal{Q}'} (V(F)\cap V(C))$ such that the last vertex has an edge of $C$ to $U$.
  See \autoref{fig:2pan} for an illustration.
  Since $\abs{\mathcal{Q}'}\ge 2$ and each cycle of $\mathcal{Q}'$ intersects $C$, there exist another cycle $U'\in \mathcal{Q}'\setminus \{U\}$  and a subpath $T=t_1t_2 \cdots t_x$ of $C$ such that
  \begin{itemize}
  \item $t_1\in V(U)$ and $t_x\in V(U')$, 
  \item $t_2, \ldots, t_{x-1}\notin \bigcup_{F\in \mathcal{Q}'} V(F)$.
  \end{itemize}
  For $i\in \{1,2,3,4,5\}$, we assign $t^U_i:=t_i$.
  
  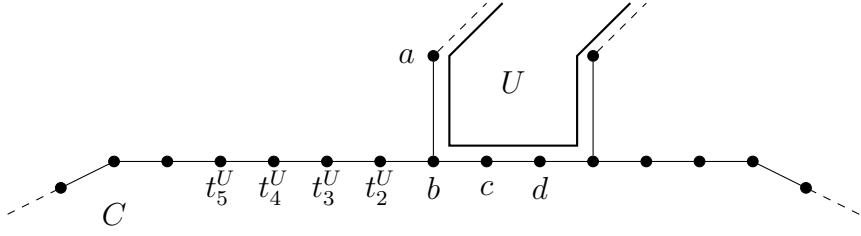
\begin{figure}
  \centering
   \begin{tikzpicture}[scale=0.7]
  \tikzstyle{w}=[circle,draw,fill=black,inner sep=0pt,minimum width=4pt]

   \draw (-2,0)--(10,0);
	\draw(-2, 0)--(-3,-0.5);
	\draw(10, 0)--(11,-0.5);
      \draw (-3,-.5) node [w] {};
       \draw (11,-.5) node [w] {};
 	\draw[dashed](12, -1)--(11,-0.5);
	\draw[dashed](-4,-1)--(-3,-0.5);
 
 \foreach \y in {-2,...,10}{
      \draw (\y,0) node [w] (a\y) {};
     }
       \draw (4,2) node [w] (x1) {};
       \draw (7,2) node [w] (x2) {};
	\draw  (x1)--(a4);
	\draw  (x2)--(a7);
	\draw[dashed] (x1)--(5,3);
	\draw[dashed] (x2)--(8,3);
	
	\draw[thick] (5.3,3)--(4.3, 2)--(4.3,0.3)--(6.7,0.3)--(6.7, 2)--(7.7,3);
	
      \node at (-2, -1) {$C$};
      \node at (5.5, 1.5) {$U$};
      \node at (4, -.5) {$b$};
      \node at (5, -.5) {$c$};
      \node at (6, -.5) {$d$};
      \node at (3.5, 2) {$a$};
          
     \node at (3, -.5) {$t^U_2$};
     \node at (2, -.5) {$t^U_3$};
     \node at (1, -.5) {$t^U_4$};
     \node at (0, -.5) {$t^U_5$};

   \end{tikzpicture}   \caption{A cycle $U$ in $\mathcal{Q}'$ and four selected vertices in \autoref{lem:2panpack}.}\label{fig:2pan}
\end{figure}
  
  We claim that $G[V(U)\cup \{t^U_i: 2\le i\le 5\}]$ contains an induced subdivision of the 2-pan.
 Let $abcd$ be the subpath of $U$ such that $b=t^U_1$ and $bc\in E(C)\cap E(U)$. 
 By \autoref{lem:dominating} and since we work in $G-D$, 
 every vertex of $U$ has at most one neighbor in $C$. In particular, $bt^U_2$ is the only edge between $\{a,b,c,d\}$ and $\{t^U_i: 2\le i\le 5\}$. Therefore, one can find an induced subdivision of the $2$-pan in $G[V(U)\cup \{t^U_i: 2\le i\le 5\}]$ using \autoref{lem:induced2pan}.

  From the choice of $\mathcal{P}'$, the sets in $\{\{t^U_i:2\le i\le 5\} : U\in \mathcal{Q}'\}$ are pairwise disjoint. Hence by applying \autoref{lem:induced2pan} as above for every $U\in \mathcal{Q}'$ we obtain in polynomial time a collection of $k$ vertex-disjoint induced subdivisions of the $2$-pan. 
  
  This concludes the proof.
\end{proof}

The last thing to show is that in fact, an induced subdivision of the $2$-pan never contains a vertex dominating~$C$.
\begin{lemma}\label{lem:avoiddom}
  Let $k$ be a positive integer, let $G$ be a graph, $(v_1, v_2 ,C)$ be an induced subdivision of the $2$-pan in $G$ with minimum number of vertices and $\abs{C}\ge 11$, and $D$ be the set of all vertices dominating $C$. 
  Every induced subdivision of the $2$-pan in $G-\{v_1, v_2\}$ has no $C$-dominating vertices.
\end{lemma}
\begin{proof}
Suppose that there exists an induced subdivision $H$ of the $2$-pan in $G-\{v_1, v_2\}$ containing a $C$-dominating vertex~$v$.
We prove two claims.

\begin{claim}\label{claim:inducedpath1}
There is no induced path $p_1p_2p_3$ in $H-v$ such that 
$p_1\in V(C)$ and $p_2, p_3\in V(G)\setminus (V(C)\cup D)$.
\end{claim}
\begin{proofclaim}
Let $q_1, q_2$ be the neighbors of $p_1$ in $C$.
As $v$ is $C$-dominating, for each $i\in \{1, 2\}$, $vp_1q_iv$ is a triangle.
By minimality of $(v_1, v_2, C)$ and the fact that $\abs{C}\ge 11$, we deduce that $q_i$ does not belong to~$H$. For the same reason, $v$ has no neighbor in $\{p_2, p_3\}$.

Let $i\in \{1, 2\}$. Consider the subgraph induced by $\{v, q_i, p_1, p_2, p_3\}$.
As above the graph $G[\{v, q_i, p_1, p_2, p_3\}]$ is not a $2$-pan. So 
$q_i$ is adjacent to either $p_2$ or $p_3$.
Since $p_2\notin D$, $p_2$ has exactly one neighbor in $C$ by \autoref{lem:dominating} (which is $p_1$) and thus it is not adjacent to~$q_i$.
Therefore, both $q_1$ and $q_2$ are adjacent to~$p_3$.
By \autoref{lem:dominating}, $p_3$ is $C$-dominating, and thus $p_3$ is adjacent to $p_1$. This contradicts the assumption that $p_1p_2p_3$ is an induced path.
We conclude that there is no such an induced path.
\end{proofclaim}

A similar observation can be made for a path continuing from a $C$-dominating vertex.

\begin{claim}\label{claim:inducedpath2}
There is no induced path $wvp_1p_2$ in $H$ such that 
$w\in V(C)$ and $p_1, p_2\in V(G)\setminus (V(C)\cup D)$.

\end{claim}
\begin{proofclaim}
Let $q_1, q_2$ be the neighbors of $w$ in $C$. As above we can deduce that none of $q_1$ and $q_2$ belongs to~$H$.
Also, for every $i\in \{1, 2\}$, $G[\{v, w, q_i, p_1, p_2\}]$ is not a $2$-pan, so
$q_i$ is adjacent to either $p_1$ or $p_2$.
If both $q_1$ and $q_2$ are adjacent to $p_1$ (or $p_2$), then 
by \autoref{lem:dominating}, 
$p_1$ (or $p_2$) is $C$-dominating, a contradiction. 
Therefore, we may assume, without loss of generality, that $q_1$ is adjacent to $p_1$, but not to $p_2$, 
and $q_2$ is adjacent to $p_2$ but not to $p_1$.
Then the cycle $p_1p_2q_2wq_1$ with two more vertices in $C$ induces a $2$-pan, which is smaller than $(v_1, v_2, C)$, a contradiction.
We conclude that there is no such an induced path.
\end{proofclaim}
Observe that $\abs{V(H)\cap V(C)}\le 3$, otherwise the $C$-dominating vertex in $H$ has degree $4$ in $H$.
Also, since $H$ contains a vertex of $C$, we have $\abs{V(H)\cap D}\le 3$.
Furthermore, $\abs{V(H)\cap V(C)}\ge 2$ and $\abs{V(H)\cap D}\ge 2$, then $H$ contains an induced subgraph isomorphic to $C_4$, a contradiction.
Thus, we may assume that $\abs{V(H)\cap V(C)}+\abs{V(H)\cap D}\le 4$.

So $H-(V(C)\cup D)$ has at least $11-4=7$ vertices divided into at most $5$ connected components.
Thus, one of connected components of $H-(V(C)\cup D)$ contains an edge $x_1x_2$ whose one end vertex, say $x_1$, has a neighbor in $V(H)\cap (V(C)\cup D)$.
We call~$x$ a neighbor of $x_1$ in $V(H)\cap (V(C)\cup D)$.

If $x\in V(C)$, then $x$ is not adjacent to $x_2$; otherwise, $H$ contains an induced subgraph isomorphic to $C_3$.
Then $xx_1x_2$ is an induced path, and this is a contradiction by \autoref{claim:inducedpath1}.
So, we may assume that $x\in D$. Note that similarly $x$ is not adjacent to $x_2$, and since $H$ does not contain $v_1, v_2$, $H$ contains a vertex of $C$, say $y$. 
Then $y$ has no neighbors in $\{x_1, x_2\}$; otherwise, $G[\{v,y,x_1, x_2\}]$ contains $C_3$ or $C_4$, a contradiction.
 But then $yvx_1x_2$ is an induced path, and this is a contradiction by \autoref{claim:inducedpath2}.

We conclude that there is no such an induced subdivision of the $2$-pan.
\end{proof}

We can now prove \autoref{t:2panep}.
\begin{proof}[Proof of \autoref{t:2panep}.]
As in the proof for 1-pans, we construct a maximal sequence of graphs $G_1,\ldots, G_{\ell+1}$ and $F_1,\ldots , F_{\ell}$ 
such that 
\begin{itemize}
\item $G_1=G$;
\item for each $i\in \{1, \ldots, \ell\}$, $F_i$ is an induced subdivision of the $2$-pan in $G_i$ with minimum number of vertices; and
\item for each $i\in \{1, \ldots, \ell\}$, $G_{i+1}=G_i-V(F_i)$.
\end{itemize}
Such a sequence can be constructed in polynomial time by repeatedly applying \autoref{lem:algorithm}.
If $\ell \geq k$, then we have found a  packing of $k$ induced subdivisions of the $2$-pan. 
Hence, we assume in the sequel that $\ell \leq k-1$. 
Let $\mu_k:=792k\log k +25k-23$.

\begin{claim}Let $j \in \intv{1}{\ell+1}$.
One can find in polynomial time either $k$ vertex-disjoint induced subdivisions of the $2$-pan, or 
a vertex set $X_{j}$ of $G_{j}$ of size at most $(\ell+1-j)\mu_k$ such that $G_j-X_j$ has no induced subdivision of the $2$-pan.   
\end{claim}

\begin{proofclaim}
  We prove claim by induction for $j=\ell+1$ down to $j=1$. 
The claim trivially holds for $j=\ell+1$ with $X_{\ell+1}=\emptyset$ because $G_{\ell+1}$ has no induced subdivision of the $2$-pan.
Let us  assume that for some $j\leq \ell$, we obtained a required vertex set $X_{j+1}$ of $G_{j+1}$ of size at most $(\ell-j)\mu_k$.
Then in $G_{j}-X_{j+1}$, $F_{j}$ is an induced subdivision of the $2$-pan with minimum number of vertices.
If $F_{j}$ has less than 13 vertices, then we set $X_{j}:=X_{j+1}\cup V(F_{j})$. Clearly, $\abs{X_j}\leq (\ell-j+1)\mu_k$.
We may assume $F_j$ has at least $13$ vertices. Let $F_j:=(u_1, u_2,C)$, and let $D$ be the set of vertices in $G_j-X_{j+1}$ that dominate $C$.
Note that $C$ has length at least $11$.

According to \autoref{lem:avoiddom}, there are no induced subdivisions of the $2$-pan intersecting $D$. 
Therefore, we can ignore the vertex set $D$.

We apply Gallai's $A$-path Theorem (\autoref{t:gallaiapath}) with $A=V(C)$ for finding at least $396k \log k$ pairwise vertex-disjoint $V(C)$-paths in $G_j-(D\cup X_{j+1}\cup \{u_1, u_2\})-E(C)$.
Assume that it outputs such $V(C)$-paths.
Then, by applying \autoref{lem:2panpack} to $G_{j}-X_{j+1}$ and $V(C)$,
one can find in polynomial time $k$ vertex-disjoint induced subdivisions of the $2$-pan.
Thus, we may assume that \autoref{t:gallaiapath} outputs a vertex set $S$ of size at most $792 k\log k$ hitting all $V(C)$-paths in $G_j-(D\cup X_{j+1}\cup \{u_1, u_2\})-E(C)$.
 
Now, we consider the graph $G_j':=G_j-(X_{j+1}\cup D\cup S\cup \{u_1, u_2\})$. 
Suppose $G_j'$ contains an induced subdivision $Q=(w_1, w_2, W)$ of the $2$-pan.
Then $G_j'[V(C)\cap V(Q)]$ is connected; otherwise, $G_j'$ contains a $V(C)$-path in $(G_j-D-X_{j+1})-E(C)$, contradicting the fact that $S$ hits all such $V(C)$-paths.
Furthermore, $G_j'[V(C)\cap V(Q)]$ contains no edge of $W$; otherwise, we also have a $V(C)$-path in $(G_j-D-X_{j+1})-E(C)$.
Thus, we have that $\abs{V(C)\cap V(Q)}\le 3$.

We recursively construct sets $U\subseteq V(C)$ and $\mathcal{J}$ as follows. We start with  $U:=\emptyset$ and $\mathcal{J}:=\emptyset$.
For every set of five consecutive vertices $v_1, v_2, v_3, v_4, v_5$ of $C$ with $v_2, v_3, v_4\notin U$, 
we test whether $G_j'[(V(G_j')\setminus V(C))\cup \{v_2,v_3, v_4\}]$ contains an induced subdivision $H$ of the $2$-pan, 
and if so, add vertices $v_1, v_2, v_3, v_4, v_5$ to $U$, and add $H$ to $\mathcal{J}$ and increase the counter by 1. If the counter reaches $k$, then we stop. 
Note that graphs in $\mathcal{J}$ are pairwise vertex-disjoint; otherwise, we can find a $V(C)$-path, a contradiction.
Thus, if the counter reaches $k$, we can output $k$ pairwise vertex-disjoint induced subdivisions of the $2$-pan.

Assume the counter does not reach $k$. Then
we have $\abs{U}\leq 5(k-1)$. We take the $2$-neighborhood $U'$ of $U$ in $C$, and observe that $\abs{U'}\le 25(k-1)$. 
We claim that $G_j'-U'$ has no induced subdivision of the $2$-pan. 
Suppose for contradiction that there is an induced subdivision $F$ of the $2$-pan in $G_j'-U_j$.
The intersection of $F$ on $C$ is a set of at most three consecutive vertices, say $T$.

Since $\dist_C(T, U)\ge 2$ by the construction of $U'$, 
there exists three consecutive vertices $z_1, z_2, z_3$ of $C$ not containing $U$,  
such that $G_j'[(V(G_j')\setminus V(C))\cup \{z_1, z_2, z_3\}]$ contains an induced subdivision of the 2-pan.
It implies that we can increase the counter by adding this induced subdivision to $\mathcal{J}$, a contradiction.
Therefore, $G_j'-U'$ has no induced subdivision of the $2$-pan.

Observe that 
\[ \abs{X_j}\le \abs{X_{j+1}\cup S\cup \{u_1, u_2\}\cup U'}\le (\ell-j)\mu_{k}+792k\log k +2+25(k-1)\le (\ell-j+1)\mu_k.\] 
So, $X_j:=X_{j+1}\cup S\cup \{u_1, u_2\}\cup U'$ satisfies the claim.
\end{proofclaim}
The result follows from the claim with $j=1$.
\end{proof}

\section{Subdivisions of the diamond have the induced \texorpdfstring{\ep{}}{Erdos-Posa} property}\label{sec:diamond}

In this section, we prove that the subdivisions of the diamond have the induced \ep{} property.

\begin{theorem}\label{t:diamond}
There exists a polynomial function $g:\mathbb{N}\rightarrow \mathbb{N}$ satisfying the following.
Given a graph $G$ and a positive integer $k$, one can in time $\mathcal{O}(kN(3, 3k)^{3k}+k^2\abs{G}^7)$
output either $k$ vertex-disjoint induced subdivisions of the diamond, or a vertex set of size at most $g(k)$ hitting every induced subdivision of the diamond.
\end{theorem}

Remark that while the bounding function $g$ that we obtained is polynomial, our upper-bound on its order is large; 
$g(k) =  \mathcal{O}(N(4,3k))$, where $N$ is the function of \autoref{prop:regularpartition}.

We follow a similar line of proofs as in the previous section.
We first deal with the case where the considered graph $G$ has an
induced subdivision $H$ of the diamond and $G-V(H)$ is $H$-induced-subdivision-free (here we do not need the minimality of a model).
Observe that a subdivision of the diamond consists of three internally
disjoint paths between two distinct vertices.
This simple observation allows us to focus on the following sightly easier setting:
given a graph $G$ and an induced path $Q$, we aim at finding
either many vertex-disjoint induced subdivisions of the diamond or a
small vertex set hitting all the induced subdivisions of the diamond
that meet~$Q$.

For a vertex subset $S$ of a graph $G$ and $v\in V(G)\setminus S$, 
a path is called a \emph{$(v,S)$-path}, if it starts with $v$ and ends at a vertex in $S$ and contains no other vertices in $S$.

We will use the \ep{} property of $A$-$\ell$-combs, recently developed
by Bruhn, Heinlein, and Joos~\cite{Bruhn2017}, to first exclude
induced diamond subdivisions of a special type.
Since we only use it with $\ell=2$, we name it \emph{$A$-claw} for simplicity. 
Given a graph $G$ and a vertex subset $A$, we say that a subgraph $H$ of $G$ is an \emph{$A$-claw} if 
\begin{itemize}
\item $H$ contains a vertex $v\in V(G)\setminus A$, 
\item $H$ consists of three $(v,A)$-paths $P_1, P_2, P_3$ such that
  $V(P_i) \cap V(P_j) = \{v\}$, for every distinct $i,j\in \{1,2,3\}$.
\end{itemize}
For convenience, if $H$ is a subgraph of $G$, then we refer to $V(H)$-claws as \emph{$H$-claws}. The \emph{leaves} of an $A$-claw are its vertices in~$A$.

\subsection{Base polynomial-time algorithms}\label{sec:basepoly}
In this section we present algorithms to find (collections of) induced subdivisions of the diamond and claws.
We first show that one can detect an induced subdivision of the diamond in polynomial time.
The following lemma is useful.

\begin{lemma}\label{lem:threepaths1}
Given a graph $G$, an induced path $Q$, and a $Q$-claw $F$,
one can find in time $\mathcal{O}(\abs{G}^3)$ an induced subdivision of the diamond in $G[V(F)\cup V(Q)]$.
\end{lemma}
\begin{proof}
  We prove by induction on $|V(F)\cup V(Q)|$ that one can find an induced subdivision of the diamond in $G[V(F)\cup V(Q)]$ in $\abs{V(F)\cup V(Q)}\cdot \mathcal{O}(\abs{E(G)})$ steps.
  
  Let $v$ be the vertex of $F - V(Q)$ that is connected via the paths $P_1, P_2, P_3$ to $Q$, forming the $Q$-claw $F$.
Without loss of generality, we assume that the end vertices of $P_1, P_2, P_3$ appear in $Q$ in this order and that $P_1$ and $P_3$ meet $Q$ on its endpoints.
For each $i\in \{1,2\}$, let $Q_i$ be the subpath of $Q$ from the end vertex of $P_i$ to the end vertex of $P_{i+1}$. 
Let $H:=G[V(F) \cup V(Q)]$.

If $H$ contains no edges other than those in $P_1\cup P_2\cup P_3\cup Q$, 
then clearly, it is an induced subdivision of the diamond.
Besides, we may assume that each of $P_1, P_2, P_3$ is an induced path; otherwise we could shorten it and apply the induction hypothesis.
Therefore, we may assume that $H$ contains an edge $wz$, which does not have
both endpoints in one of $P_1$, $P_2$, $P_3$, or~$Q$.

Suppose $w$ is an internal vertex of $P_1$ and $z$ is an internal vertex of  $P_2\cup Q_1$.
Then from $z$, there are three paths to $P_1$, where the end vertices of two paths are the end vertices of $P_1$. 
As $P_1$ is an induced path, by induction hypothesis, one can find an
induced subdivision of the diamond in $G[V(P_1\cup P_2\cup Q_1)]$ in
time $\abs{V(P_1\cup P_2\cup Q_1)} \cdot \mathcal{O}(E(G))$.

We thus may assume that there is no such an edge. 
We can apply the same argument for all pairs $(P_2, P_1\cup Q_1), (P_2, P_3\cup Q_2), (P_3, P_2\cup Q_2)$.

  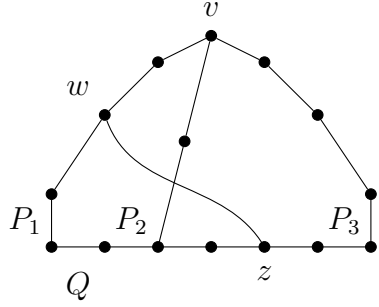
\begin{figure}
  \centering
   \begin{tikzpicture}[scale=0.7]
  \tikzstyle{w}=[circle,draw,fill=black,inner sep=0pt,minimum width=4pt]

  \draw (1,0)--(7,0);
  \foreach \y in {1,...,7}{
      \draw (\y,0) node [w] (a\y) {};
     }
       \draw (4,4) node [w] (x1) {};
       \draw (3,3.5) node [w] (x2) {};
       \draw (2,2.5) node [w] (x3) {};
       \draw (1,1) node [w] (x4) {};
     \draw(x1)--(x2)--(x3)--(x4)--(1,0);
     
     \draw (3.5, 2) node [w](y1) {};
     \draw(x1)--(y1)--(3,0);
     
       \draw (5,3.5) node [w] (z2) {};
       \draw (6,2.5) node [w] (z3) {};
       \draw (7,1) node [w] (z4) {};
     \draw(x1)--(z2)--(z3)--(z4)--(7,0);
     
      \node at (1.5, -0.75) {$Q$};
      \node at (.5, .5) {$P_1$};
      \node at (2.5, .5) {$P_2$};
      \node at (6.5, .5) {$P_3$};

      \draw(x3) [in=120,out=-70] to (5,0);   
      \node at (4, 4.5) {$v$};
      \node at (1.5, 3) {$w$};
      \node at (5,  -.5) {$z$};

   \end{tikzpicture}   \caption{An edge $wz$ where $w$ is an internal vertex of $P_1$ and $z$ is an internal vertex of $Q_2\cup P_3$ in \autoref{lem:threepaths1}.}\label{fig:crossingedge}
\end{figure}

Since $Q$ is an induced path, one of $w$ and $z$ is contained in $F - V(Q)$.
If $w\in V(P_2)\setminus V(Q)$, then $z$ is contained in $P_1\cup P_3\cup Q-V(P_2)$, which is not possible by the previous argument.
Therefore, we may assume that $w$ is an internal vertex of $P_1$ without loss of generality, and 
$z$ is an internal vertex of $P_3\cup Q_2$.
See \autoref{fig:crossingedge} for an illustration.

Let $r$ be the neighbor of $v$ in $P_2$, 
and let $P_1'$ be the subpath of $P_1$ from $v$ to $w$.
Note that $G[V(P_1')\cup \{r\}]$ is an induced path, and there are three paths from $z$ to $P_1'$, 
namely, two paths along $P_2\cup P_3\cup Q_2$ and $wz$.
Since the union of those paths does not contain the end vertex of $P_1$ in $Q$, we obtained a smaller induced subgraph graph satisfying the premisses of the lemma. 
By the induction hypothesis, one can find an induced subdivision of
the diamond in time $\abs{V(P_1'\cup P_2\cup P_3\cup Q_2)}\cdot \mathcal{O}(E(G))$.

This completes the proof.
\end{proof}

\begin{proposition}\label{prop:detectdia}
Given a graph $G$, one can test in time $\mathcal{O}(\abs{G}^7)$ whether $G$ contains an induced subdivision of the diamond, and output one if exists.
\end{proposition}
\begin{proof}
Let us consider 4 vertices $a,b,c$, and $d$ such that
$abc$ is an induced path and $d$ is adjacent to neither $a$ nor $c$.

Using Menger's theorem, we can check in $\mathcal{O}(\abs{G}^3)$-time whether there are three internally vertex-disjoint paths from $d$ to $a,b,c$.
If such paths do not exist, then there is no induced subdivision $H$ such that
\begin{itemize}
\item $b, d$ are vertices of degree 3 in $H$, and
\item $a, c$ are neighbors of $b$ in $H$.
\end{itemize}
On the other hand, if the 3 paths exist, then by
\autoref{lem:threepaths1}, we can detect such an induced subdivision of the
diamond in $\mathcal{O}(\abs{G}^3)$ steps.
By iterating over all possible choices of vertices $a,b,c$, and $d$,
we get a total running time is $\mathcal{O}(\abs{G}^7)$.
\end{proof}

Now we prove that every large enough collection of $Q$-claws contains many disjoint  induced subdivisions of the diamond.
\begin{lemma}\label{lem:clawtodiamond}
Let $G$ be a graph and let $Q$ be an induced path of $G$.
  Given a set of $N(3, 3k)$ vertex-disjoint $Q$-claws, one can find $k$ vertex-disjoint induced subdivisions of the diamond in time  $\mathcal{O}(N(3, 3k)^{3k}+\abs{G}^3)$.
\end{lemma}
\begin{proof}
  Let $Q=v_1v_2 \cdots v_m$, and $I=\{1, \ldots, m\}$.
Let $T_1, T_2, \ldots, T_{N(3, 3k)}$ be a given set of pairwise vertex-disjoint $Q$-claws.
For each $i\in \{1, \ldots, N(3, 3k)\}$, let $A_i:=\{j\in \mathbb{N}:v_j\in V(T_i)\}$.
By definition, $\abs{A_i}=3$.

We apply the regular partition lemma (\autoref{prop:regularpartition}) to $A_i$'s  with $n=3$.
Then there exist a subsequence $(A_{c_1}, \ldots, A_{c_{3k}})$ of $(A_1, \ldots, A_{N(3, 3k)})$ and a regular partition of $I$ with respect to $(A_{c_1}, \ldots, A_{c_{3k}})$ that has order at most $3$. 
Since given $Q$-claws are pairwise vertex-disjoint, in each part, 
restrictions of $A_{c_1}, \ldots, A_{c_{3k}}$ are not the same (i.e. item \eqref{it:rp1} of the definition of a regular partition does not hold).
Algorithmically, the regular partition can be found in $\mathcal{O}(N(3, 3k)^{3k})$-time by simply iterating over possible $3k$-tuples until one forming a regular partition is found.

Suppose first that the order of the partition is $1$. 
Then one of the following holds:
\begin{itemize}
\item For all $j,j'\in \{1, \ldots, 3k\}$ with $j<j'$, $\max(A_{c_j})<\min (A_{c_{j'}})$.
\item For all $j, j'\in \{1, \ldots, 3k\}$ with $j<j'$, $\max (A_{c_{j'}})< \min (A_{c_j})$. 
\end{itemize}
For each $i\in \{1, \ldots, 3k\}$, let $Q_i$ be the minimal subpath of $Q$ containing the vertices in $\{v_j:j\in A_{c_i}\}$.
Observe that the $Q_i$'s are disjoint and that $T_{c_i}$ is a $Q_i$-claw, for every $i\in \{1, \ldots, 3k\}$.
Using \autoref{lem:threepaths1}, one can find in time $\mathcal{O}(\abs{G}^3)$ an induced subdivision of the diamond in each $G[V(T_{c_i})\cup V(Q_i)]$, so we are done.

Suppose the order of the partition is $2$ or $3$. 
Let $I_1$ and $I_2$ be the first two parts in the partition, and we may assume that $A_{c_i}$ has an element in each of the two parts. 
For each $i\in \{1, \ldots, 3k-2\}$, let $S_i$ be the minimal subpath of $Q$ containing the vertices in 
\[\{v_j:j\in (A_{c_i}\cup A_{c_{i+1}}\cup A_{c_{i+2}})\cap I_1\},\] 
and let $R_i$ be the minimal subpath of $Q$ containing the vertices in 
\[\{v_j:j\in (A_{c_i}\cup A_{c_{i+1}}\cup A_{c_{i+2}})\cap I_2\}.\] 
Then subgraphs in 
\[\{ S_{3\ell-2}\cup R_{3\ell-2} \cup T_{c_{3\ell-2}}\cup T_{c_{3\ell-1}}\cup T_{c_{3\ell}} :  1\le \ell\le k\}\] are pairwise vertex-disjoint.

Furthermore, since $S_{3\ell-2}$ is an induced path, and there are three internally vertex-disjoint paths from the intersection of $R_{3\ell-2}$ and $T_{c_{3\ell-1}}$
to $S_{3\ell-2}$ in $T_{c_{3\ell-2}}\cup T_{c_{3\ell-1}}\cup T_{C_{3\ell}}\cup S_{3\ell-2}\cup R_{3\ell-2}$, 
one can output in time $\mathcal{O}(\abs{G}^3)$ an induced subdivision of the diamond in the subgraph, using \autoref{lem:threepaths1}.
Therefore, 
we obtain $k$ pairwise vertex-disjoint induced subdivisions of the diamond, as required.
\end{proof}

In \cite{Bruhn2017}, Bruhn, Heinlein, and Joos show that $A$-claws
have the (non-induced) \ep{} property.
Their proof is unfortunately not written algorithmically, however it
can easily be turned into a polynomial-time algorithm.
For completeness, we give an algorithmic version hereafter.

\begin{lemma}\label{lem:numberofleaves}
Let $T$ be a tree with no vertices of degree~$2$.
If $L$ is the set of all leaves of $T$, 
then $\abs{L}\ge \abs{V(T)\setminus  L}-2$.
\end{lemma}
\begin{proof}
We note that
\begin{itemize}
\item $\abs{L}+\abs{V(T)\setminus  L}=\abs{T}=\abs{E(T)}-1$, and
\item $\abs{L}+ 3\abs{V(T)\setminus  L}\le \sum_{t\in V(T)}\deg_T(t) = 2\abs{E(T)}$.
\end{itemize}
Combining the two equations, we have that 
$\abs{A}\ge \abs{V(T)\setminus L}-2$, as required.
\end{proof}

\begin{lemma}\label{lem:countleaves}
Let $F$ be a forest of maximum degree 3 and where each component has a
degree-3 vertex.
Let $L$ be the set of its leaves, and $k$ be a positive integer.
\begin{enumerate}

\item \label{it:claws} If $F$ has at least $6k$ leaves, then in time
  $\mathcal{O}(k\abs{F})$ one can find $k$ pairwise vertex-disjoint $L$-claws in~$F$.
\item \label{it:clcov} If $F$ contains less than $6k$ leaves and less than $k$ connected components, then $F$ contains less than $14k$ vertices of degree $1$ or $3$.
\end{enumerate}
\end{lemma}

\begin{proof}
  Proof of \eqref{it:claws}.
In this proof we call \emph{good forest} every forest of maximum degree 3 and where each component has a
degree-3 vertex.
  We prove by induction on $k$. If $k=1$,
then there is at least one component containing an $L$-claw and
\eqref{it:claws} holds. So, we may assume that $k\ge 2$ and suppose
that the statement holds for smaller values of~$k$.
Observe if $F$ has a vertex of degree 2, one can delete it and
add an edge between its neighbors without changing the number of
leaves neither the existence of $L$-claws. Therefore we may assume
that $F$ has no degree-2 vertex.

Let us fix a root node in $F$ and call $T$ the component containing~it.
We take a furthest node $v$ from the root that is not a leaf.
By our assumption above, it has exactly two children.
Such a node can be found in linear time using Breadth First Search.
If $v$ is the root of $T$, then $T$ is an $L$-claw, 
and since $T$ has maximum degree 3, the remaining part of $F$ contains
at least $6k-3$ leaves.
Clearly $F-V(T)$ is a good forest.
Thus, by the induction hypothesis, we can obtain $k-1$ pairwise
vertex-disjoint $L$-claws in $F-V(T)$ in time $\mathcal{O}((k-1)
\abs{F})$. Together with $T$, they form a collection of $k$
vertex-disjoint $L$-claws in~$F$, that we found in time
$\mathcal{O}(k \abs{F})$.
Thus in the sequel we may assume that $v$ is not the root.

Let $w$ be the parent $v$ and let $T_w$ be the subtree of $T$ that is
induced by the set of all descendants of $w$ (including~$w$). Observe
that $T_w$ has at least 3 leaves, hence it contains an
$L$-claw. This claw can be found in linear time using Breadth First
Search.

We distinguish two cases:
\begin{enumerate}[{Case}~1:]
\item $w$ is the root. Then $T = T_w$ has at most 6 leaves. Clearly $F - T$ is a good forest with at least $6k - 6$ leaves.
\item $w$ is not the root. Then $T_w$ has at most 4 leaves.  Let $x$ be the vertex of $T - V(T_w)$ at minimum distance from
$v$. Let $P$ be the path of $T$ from $x$ to $w$ and let $T^- = T -(V(T_w) \cup V(P) \setminus
\{x\})$. Then $T^-$ is a good forest with at least $6k - 6$ leaves.
\end{enumerate}

In each case $F - T_w$ contains a good subforest with at least
$6(k-1)$ leaves (each belonging to $L$) and which can be found in
linear time.
Considering the aforementioned claw present in $T_w$ and applying the induction
hypothesis to $F-T_w$, we finally obtain as above $k$ pairwise
vertex-disjoint $L$-claws in time $\mathcal{O}(k\abs{T})$, as required.

\medskip

Proof of \eqref{it:clcov}. Let $n$ be the number of components of $F$ . By the assumption, we have $n<k$.
Let $T_1, T_2, \ldots, T_n$ be the connected components of $F$, 
and for each $i\in \{1, 2, \ldots, n\}$, 
let $m_i$ be the number of vertices of degree 3 in $T_i$
and let $\ell_i$ be the number of vertices of degree 1 in $T_i$.
By \autoref{lem:numberofleaves},
we know that $\ell_i\ge m_i-2$.
Since $\sum_{i=1}^n \ell_i< 6k$,
we have $\sum_{i = 1}^n m_i\le \sum_{i = 1}^n \ell_i+2n
< 6k + 2k = 8k$.
Thus, we have $\sum_{i = 1}^n (\ell_i + m_i) < 14k$.
\end{proof}

\begin{proposition}\label{prop:aclawalgo}
Given a graph $G$, a subset $A\subseteq V(G)$, and a positive integer $k$,
one can in time $\mathcal{O}(\abs{G}^4)$ output either $k$ pairwise vertex-disjoint $A$-claws, or a vertex set $S$ of size $14k$ hitting all $A$-claws.
\end{proposition}

\begin{proof}
Let $S_0:=\emptyset$, and $F_0$ be the empty graph.
Let us apply the following algorithm.
We start with $i=1$ and, while $G-S_{i-1}$ contains an $A$-claw, we do the following:
\begin{enumerate}
\item \label{it:findclaw} let $X$ be an $A$-claw of $G-S_{i-1}$;
\item if $X$ intersects $F_{i-1}-S_{i-1}$, define $F_i$ as the union of $F_{i-1}$ and a path from $F_{i-1}-S_{i-1}$ to $A$ (that is contained in $X$);
\item otherwise, define $F_i$ as the disjoint union of $X$ and $F_{i-1}$;
\item let $S_i$ be the set of all vertices of degree $1$ or $3$ in $F_i$;
\item increment $i$.
\end{enumerate}

Let $n$ be the maximum value for which $F_i$ and $S_i$ are defined.
Note that one can find an $A$-claw by guessing the vertex of degree
$3$ and then testing whether there are three paths from it to $A$
using Menger's theorem. Thus, checking the condition of the while loop and step
\eqref{it:findclaw} can be performed in time $\mathcal{O}(\abs{G}^3)$, and
we may construct the sequences $F_1, \dots, F_n$ and $S_1, \dots, S_n$ in time $\mathcal{O}(\abs{G}^4)$.
Let $F:=F_n$ and $S:=S_n$. By construction, we know that $G-S$ has no
$A$-claws and $F$ is a forest with all its leaves in $A$, maximum
degree~3, and where every component has a degree-3 vertex.
If $\abs{S}\le 14k$, then we are done.
We may assume that $\abs{S} > 14k$. In particular $\abs{F} > 14k$.

If $F$ contains $k$ connected components, then we can simply find one $A$-claw from each connected component, and output $k$ vertex-disjoint $A$-claws.
We may assume that $F$ contains less than $k$ connected components.
Then by \eqref{it:clcov} of \autoref{lem:countleaves}, $F$ contains at least $6k$ leaves, 
and by \eqref{it:claws} of the same lemma, one can find $k$
vertex-disjoint $A$-claws in time $\mathcal{O}(k\abs{G}) = \mathcal{O}(\abs{G}^2)$.
\end{proof}

\subsection{Structural lemmas}

For a set $A\subseteq V(G)$, a \emph{Tutte bridge} of $A$ in $G$ is a subgraph of $G$ consisting of one component $C$ of $G-A$ and all edges joining $C$ and $A$ and all vertices of $A$ incident with those edges. We discuss in this section under which conditions a Tutte bridge can be used to construct an induced subdivision of the diamond.

\begin{lemma}\label{lem:bridge3legs}
Let $Q$ be an induced path in a graph $G$ and $H$ be a Tutte bridge of $V(Q)$ in $G$ such that 
$\abs{V(H)\cap V(Q)}\ge 3$. Then $H$ contains a $Q$-claw.
\end{lemma}
\begin{proof}
Let $a,b,c$ be three distinct vertices in $V(H)\cap V(Q)$, and let $a',b',c'$ be their (not necessarily distinct) neighbors in $H-V(Q)$, respectively. 
It is easy to check that a subgraph-minimal tree of $H-V(Q)$ spanning
$a'$, $b'$, and $c'$
form (together with $a,b$, and $c$) a $Q$-claw. Such a subgraph exists
since $H-V(Q)$ is connected. 
\end{proof}

Recall that our intermediate goal is to prove an induced \ep{} type result
for subdivisions of the diamond intersecting a given path $Q$ of a
graph $G$. We saw in the previous section that $Q$-claws can be used
to construct induced subdivisions of the diamond (\autoref{lem:threepaths1}), and how to deal with
those (\autoref{prop:aclawalgo}).
So we may now focus on the case where there is no $A$-claw. In particular,
\autoref{lem:bridge3legs} allows us to assume that there is no Tutte
bridge $H$ of $V(Q)$ such that $\abs{V(H)\cap V(Q)}\ge 3$.
The following lemma shows that if a Tutte bridge containing a cycle
$C$ connected to $Q$ via two disjoint paths (with some additional properties), then
one can also find an induced subdivision of the diamond. Note that first two conditions do not always imply the existence of an induced subdivision of the diamond: if $ab$ is an edge that is a component of $G-V(Q)$ and $N(a)\cap V(Q)=N(b)\cap V(Q)=\{p,q\}$ for some two consecutive vertices $p,q$ in $Q$, then this Tutte bridge induces a $K_4$, that does not contain an induced subdivision of the diamond. 
To avoid this, we require that the two vertices in $V(H)\cap V(Q)$ are not consecutive.

\begin{lemma}\label{lem:connectedcycle}
Let $Q$ be an induced path in a graph $G$ and $H$ be a Tutte bridge of
$V(Q)$ in $G$ such that:
\begin{enumerate}
\item $V(H)\cap V(Q)=\{v_1, v_2\}$ for some non-adjacent vertices
  $v_1, v_2$ in $Q$; and
\item \label{it:cycleC} some cycle of $H$ is connected to $\{v_1, v_2\}$ via two vertex-disjoint (possibly empty) paths,
\end{enumerate}
then $G[V(H)\cup V(Q')]$ contains an induced subdivision of the diamond, where $Q'$ is the subpath of $Q$ from $v_1$ to $v_2$.
\end{lemma}

\begin{proof}
Towards a contradiction, we assume that the statement does not hold
and we consider, among all graphs $G$, all induced paths $Q$ of $G$ and all Tutte
bridges $H$ of $V(Q)$, a triple $(G,Q,H)$ such that $H$ has minimum number
of vertices.

Let $C$ be a cycle of $H$ as in the statement and let $P_1$ and $P_2$ be the two paths from $V(H)\cap V(Q)$ to $C$
such that $v_1\in V(P_1)$ and $v_2\in V(P_2)$.
Recall that $v_1$ is not adjacent to $v_2$.
By minimality, we may assume that $V(H)=V(C)\cup V(P_1)\cup V(P_2)$ and that $P_1$ and $P_2$ are induced, otherwise we could remove vertices or take shorter paths.
Let $w_1$ and $w_2$ be the end vertices of $P_1$ and $P_2$ on $C$, respectively, and
let $Q_1$ and $Q_2$ be the two subpaths from $w_1$ to $w_2$ in $C$.

If the subgraph $C \cup P_1 \cup P_2$ is induced in $G$, then $G[V(H)\cup V(Q')]$ is an induced subdivision of the diamond and we are done.
So we now consider all possible ways this subgraph can be non-induced.

We first show that $C$ is an induced cycle.
If $w_1$ is adjacent to $w_2$, then $C$ is induced; because otherwise 
we could take an induced cycle in $G[V(C)]$ containing the edge
$w_1w_2$, which would be shorter than $C$, a contradiction.
In the case where $w_1w_2\notin E(G)$, we may assume that each of $Q_1$ and $Q_2$ is induced,
otherwise, we could find a shorter cycle.
Now, observe than if $C$ has a chord $e$ from an internal vertex $z_1$ of $Q_1$ to an internal vertex $z_2$ of $Q_2$, then $Q_1$ together with $e$ forms a $Q_2$-claw and thus contain an induced subdivision of the diamond, according to \autoref{lem:threepaths1}.
Therefore $C$ is indeed induced.

Observe that there are no edges between $H-V(Q)$ and $Q-\{v_1, v_2\}$, because of the condition that $\abs{V(H)\cap V(Q)}=2$.
Also, $P_i$ has no neighbors in $P_{3-i}-w_{3-i}$ for each $i \in \{1,2\}$, otherwise we could find a smaller cycle such as $C$.
Therefore, we may assume that there is an edge between some vertex $y\in V(P_i-w_i)$ 
and some internal vertex of $Q_j$, for some $i,j\in \{1,2\}$.
Then $H$ contains a $P_i$-claw constructed as follows. The degree-3 vertex of the claw is $y$; it is connected to $P_i$ using the aforementioned edge (first path), the subpath of $Q_j$ connecting $y$ to $w_1$ (second path) and the rest of $Q_j$ together with $P_{3-i}$ and $Q'$ (third path).
By to \autoref{lem:threepaths1} $G[V(H) \cup V(Q')]$ then contains an induced subdivision of the diamond, as required.
\end{proof}

Te following lemma handles a remaining case that is not covered by \autoref{lem:connectedcycle}.

\begin{lemma}\label{lem:acyclicbridge}
Let $Q$ be an induced path in a graph $G$ and $H_1, H_2$ be two Tutte bridges of $V(Q)$ in $G$ such that:
\begin{itemize}
\item $\abs{V(H_i)\cap V(Q)}=2$ for each $i$,
\item $Q_1$ and $Q_2$ share at least one edge, where $Q_i$ is the minimal subpath of $Q$ containing $V(H_i)\cap V(Q)$, for every $i\in \{1,2\}$.
\end{itemize}
Then $G[V(H_1)\cup V(H_2)\cup V(Q_1)\cup V(Q_2)]$ contains an induced subdivision of the diamond.
\end{lemma}
\begin{proof}
Let $Q=v_1v_2 \cdots v_m$, and let $v_{a_i}, v_{b_i}$ be the end vertices of $Q_i$ such that $a_i<b_i$,  for every $i\in \{1,2\}$.
Without loss of generality, we may assume that $a_1\le a_2$.
As $Q_1$ and $Q_2$ share an edge, $b_1\ge a_2+1$.

Let $x_1$ and $y_1$ be (possibly identical) neighbors of $v_{a_1}$ and $v_{b_1}$ in $H_1-V(Q)$, respectively,
and let $R_1$ be a path from $x_1$ to $y_1$ in $H_1-V(Q)$.
Similarly, let $x_2$ and $y_2$ be neighbors of $v_{a_2}$ and $v_{b_2}$ in $H_2-V(Q)$, respectively, 
and let $R_2$ be a shortest path from $x_2$ to $y_2$ in $H_2-V(Q)$.
Let $R$ be the subpath of $Q$ from $v_{a_1}$ to $v_{a_2}$.

Observe that $G[V(R)\cup \{x_1, x_2\}]$ is an induced path from $x_1$ to $x_2$, 
because $V(H_i)\cap V(Q)$ is exactly $\{a_i, b_i\}$ for every $i\in \{1,2\}$.
Let $j=\min \{b_1, b_2\}$. It is easy to see that there are three paths from $v_j$ to the induced path $G[V(R)\cup \{x_1, x_2\}]$, 
namely, a subpath of $Q$ from $v_{a_2}$ to $v_j$, and two paths along $R_1$ and $R_2$.
Thus, by \autoref{lem:threepaths1}, 
$G[V(R_1)\cup V(R_2)\cup V(Q_1)\cup V(Q_2)]$ contains an induced subdivision of the diamond.
\end{proof}

\subsection{The main proof}
We can now describe the main proof of this section.
The following proposition asserts that the subdivisions of the diamond intersecting a given induced path have the induced \ep{} property.

\begin{lemma}\label{lem:twodisjointpaths}
Let $G$ be a connected graph and $v,w\in V(G)$ be non-adjacent vertices such that $G-\{v,w\}$ is connected.
One can in time $\mathcal{O}(|G|^2)$ whether there is a cycle $C$ with two vertex-disjoint paths from $\{v,w\}$ to $C$.
\end{lemma}
\begin{proof}
We take a block-cut decomposition of $G$. 
Let $B_v$ and $B_w$ be the blocks of $G$ containing $v$ and $w$, respectively.
Let $B_v=B_1-B_2- \cdots - B_m=B_w$ be the sequence of blocks of $G$ in the block-cut decomposition of $G$.

Suppose there is $i\in \{1, \ldots, m\}$ such that $B_i$ contains a cycle $C$.
Note that for any two vertices $a,b$ in $B_i$, there are two vertex-disjoint paths from $C$ to $\{a,b\}$, as $B_i$ is $2$-connected.
If $B_v=B_w$, then there are two vertex-disjoint paths from $C$ to $\{v,w\}$ directly.
Otherwise, along the cut vertices connecting blocks of $B_1, \ldots, B_m$, we may find two vertex-disjoint paths from $C$ to $\{v,w\}$ in $G$.

Thus, we may assume that there is no $i\in \{1, \ldots, m\}$ such that $B_i$ contains a cycle.
In this case, for every cycle $C$, there is a cut vertex separating $C$ and $\{v,w\}$.
So, we can deduce that there is no cycle $C$ with two vertex-disjoint paths from $\{v,w\}$ to $C$.
\end{proof}

\begin{proposition}\label{prop:basedpath}
  There exists a polynomial function $g_2:\mathbb{N}\rightarrow \mathbb{N}$ satisfying the following.
Given a graph $G$, an induced path $P$ of $G$, and a positive integer $k$, one can in time $\mathcal{O}(N(3, 3k)^{3k}+k\abs{G}^7)$ output either $k$ vertex-disjoint induced subdivisions of the diamond, or 
a vertex set of size at most $g_2(k)$ hitting all the induced subdivisions of the diamond that intersect~$P$.
\end{proposition}

\begin{proof}
Let $P=v_1v_2 \cdots v_m$, and $I=\{1, \ldots, m\}$.

We first apply the $A$-claw lemma to $P$. 
By \autoref{prop:aclawalgo}, one can in time $\mathcal{O}(\abs{G}^4)$ output either  
$N(3, 3k)$ pairwise vertex-disjoint $P$-claws, or a vertex set of size at most $14 N(3,3k)$ hitting all $P$-claws. 
In the former case, we use \autoref{lem:clawtodiamond} to obtain $k$ pairwise vertex-disjoint induced subdivisions of the diamond in time $\mathcal{O}(N(3, 3k)^{3k} + |G|^3)$ and we are done.

So we may assume that $G$ contains a vertex subset $X_1$ of size at most $14N(3, 3k)$
such that $G-X_1$ has no $P$-claws.
Let $G_1:=G-X_1$ and $P_1=P-X_1$.
By \autoref{lem:bridge3legs},
$G_1$ contains no Tutte bridge of $V(P_1)$ such that $\abs{V(H)\cap V(P_1)}\ge 3$.

Now, we greedily construct a maximal set $\mathcal{U}$ of pairwise vertex-disjoint Tutte bridges $H$ of $V(P_1)$ in $G_1$ such that:
\begin{itemize}
\item $\abs{V(H)\cap V(P_1)}=2$, 
\item $V(H)\cap V(P_1)$ are not consecutive vertices of $P$, and 
\item $H$ contains a cycle $C$ and two vertex-disjoint paths from $C$ to $V(H)\cap V(P_1)$ in $H$.
\end{itemize}
Since there are at most $\abs{G}$ connected components of $G_1-V(P_1)$, 
one can find such a set by considering each connected component of $G_1-V(P_1)$ and then testing whether the corresponding Tutte bridge
satisfies the conditions. The last condition can be checked in time $\mathcal{O}(|H|^2)$ using \autoref{lem:twodisjointpaths}. 

Suppose $\abs{\mathcal{U}}\ge N(2, 3k)$.
In this case, we apply the regular partition lemma (\autoref{prop:regularpartition}) with $n=2$ and obtain $k$ vertex-disjoint induced subdivisions of the diamond, with the whole path $P$. 
Following the same line of proof as in  \autoref{lem:clawtodiamond}, 
if the order of the resulting partition is $2$, 
then we can output  $k$ pairwise vertex-disjoint induced subdivisions of the diamond in time $\mathcal{O}(\abs{G}^3)$.
When the order of the partition is $1$, 
\autoref{lem:connectedcycle} implies that there is an induced subdivision of the diamond in $H$ together with the minimal subpath of $P$ containing $V(H)\cap V(P)$, for every Tutte bridge $H$ of the subset of $\mathcal{U}$ given by the regular partition lemma.
So in this case, using \autoref{prop:detectdia}, one can construct $k$ vertex-disjoint induced subdivisions of the diamond in time $\mathcal{O}(k\abs{G}^7)$.

Otherwise, let $X_2:=\bigcup_{H\in \mathcal{U}}(V(H)\cap V(P))$. 
Then we have that $\abs{X_2}\le 2N(2, 3k)$ and $X_2$ hits all Tutte bridges of $V(P_1)$ in $G_1$ satisfying the three conditions above.
Let $G_2:=G_1-X_2$ and $P_2:=P_1-X_2$.

In the next step, we greedily build a maximal set $\mathcal{W}$ of pairwise vertex-disjoint pairs of Tutte bridges $(H_1, H_2)$ of $V(P_2)$ in $G_2$ such that:
\begin{itemize}
\item $\abs{V(H_i)\cap V(P_2)}=2$, 
\item $Q_1$ and $Q_2$ share an edge, where $Q_i$ is a minimal subpath of $P$ containing $V(H_i)\cap V(P)$. 
\end{itemize}
We can construct $\mathcal{W}$ by considering all pairs of connected components of $G_2-V(P_2)$.

Suppose $\abs{\mathcal{W}}\ge 4N(4, 3k)$.
Let $v_{a_i}, v_{b_i}$ be the vertices of $V(H_i)\cap V(P)$ such that $a_i<b_i$. 
Note that $H_1$ might intersect $H_2$, and therefore, there are four types of a pair $(H_1, H_2)$, 
depending on whether $v_{a_1}=v_{a_2}$ and $v_{b_1}=v_{b_2}$.
As $\abs{\mathcal{W}}\ge 4N(4, 3k)$, there is a subset $\mathcal{W}_1$ of $\mathcal{W}$ of size at least $N(4, 3k)$
which consists of pairs of the same type.

We apply the regular partition lemma with $n$ equal to the size of $(V(H_1)\cup V(H_2))\cap V(P)$ for pairs in $\mathcal{W}_1$. 
Similar to the previous case, when the order of the partition is larger than $1$, 
then by the same line of proofs as in  \autoref{lem:clawtodiamond}, one can find in polynomial time $k$ pairwise vertex-disjoint induced subdivisions of the diamond.
When the order of the partition is $1$, \autoref{lem:acyclicbridge} implies that $H_1$, $H_2$, together with the minimal subpath of $P$ containing $(V(H_1)\cup V(H_2))\cap V(P)$ contains an induced subdivision of the diamond. 
Thus, using \autoref{prop:detectdia}, one can construct $k$ vertex-disjoint induced subdivisions of the diamond in time $\mathcal{O}(k\abs{G}^7)$.

Otherwise, let $X_3:=\bigcup_{(H_1, H_2)\in \mathcal{W}}((V(H_1)\cup V(H_2))\cap V(P))$. 
Then we have that $\abs{X_3}\le 16N(4, 3k)$ and $X_3$ hits all pairs of Tutte-bridges satisfying the above conditions. 
Let $G_3:=G_2-X_3$ and $P_3:=P_2-X_3$.

Now, we claim that the remaining induced subdivisions of the diamond have restricted positions.

\begin{claim}\label{claim:remaningdia}
Let $X$ be an induced subdivision of the diamond in $G_3$ that intersects $P_3$.
Then we have 
\begin{itemize}
	\item $X-V(P_3)$ has one component and $\abs{V(X)\cap V(P_3)}\le 2$, and 
	\item if $\abs{V(X)\cap V(P_3)}=2$, then the two vertices in $V(X)\cap V(P_3)$ are adjacent.
	\end{itemize}
\end{claim}
\begin{proofclaim}
  Suppose that $\abs{V(X)\cap V(P_3)}>1$. 
  Observe that if $X-V(P_3)$ intersects some Tutte bridge $F$ of $V(P_3)$, 
  then $\abs{V(F)\cap V(P_3)}=2$. The upper-bound holds because
  otherwise $H$ would have been considered when constructing $X_1$, while the lower-bound holds as we assume
  $\abs{V(X)\cap V(P_3)}>1$, and $X$ is connected.

  Also, if $X-V(P_3)$ intersects exactly one Tutte bridge $F$ of $V(P_3)$, 
  then $F$ contains a cycle where there are two vertex-disjoint paths from the cycle to $V(F)\cap V(P_3)$.
  So, in that case, the two vertices of $V(F)\cap V(P_3)$ are
  consecutive (otherwise we would have considered $F$ when
  constructing~$X_2$) and we are done.

  Thus, we may assume that there are at least two Tutte bridges of
  $V(P_3)$ intersecting $X$ outside $P_3$. We obtain a contradiction as follows. Let $F_1, F_2, \ldots, F_n$ be the set of Tutte bridges of $V(P_3)$ such that 
  $F_i-V(P_3)$ contains a vertex of $X$.
  For each $i
  \in \intv{1}{n}$, let $Q_i$ be the minimal subpath of $P_3$ containing the two vertices of $V(F_i)\cap V(P_3)$.
  Observe that no two such paths share an edge, because the corresponding bridges have been handled when constructing $X_3$.
  On the other hand if no two paths in $Q_1, \ldots, Q_n$ share an edge, 
  then it is easy to see that $X$ has a cut vertex, a
  contradiction.
\end{proofclaim}

By \autoref{claim:remaningdia} if $G_3$ contains an induced subdivision $X$ of the diamond, then 
$V(X)\cap V(P)$ consists of at most two consecutive vertices of $P$.
Therefore, we can find in polynomial time either $k$ pairwise vertex-disjoint induced subdivisions of the diamond, 
or a vertex set $X_4$ of size at most $2k$ hitting all remaining induced subdivisions of the diamond.
In the latter case, we obtain a hitting set of size at most 
\[\abs{X_1\cup X_2\cup X_3\cup X_4}\le 14 N(3, 3k)+N(2, 3k)+16N(4, 3k)+2k.\]
So, the function $g_2(k)=14 N(3, 3k)+N(2, 3k)+16N(4, 3k)+2k$ satisfies the statement.
\end{proof}

	Now, we prove the second intermediate proposition.
\begin{proposition}\label{prop:basedpath2}
  There exists a polynomial function $g_1:\mathbb{N}\rightarrow \mathbb{N}$ satisfying the following. 
  Given a graph $G$, an induced subdivision $H$ of the diamond in $G$, and a positive integer $k$ such that 
  \begin{itemize}
  \item $G-V(H)$ has no induced subdivision of the diamond,
  \end{itemize}
  then one can in time $\mathcal{O}(N(3, 3k)^{3k}+k\abs{G}^7)$ output either $k$ vertex-disjoint induced subdivisions of the diamond, or a vertex set of size at most $g_1(k)$ hitting every induced subdivision of the diamond.
\end{proposition}
\begin{proof}
We set $g_1(k):=3g_2(k)$.
Let $P_1, P_2, P_3$ be the three paths forming $H$.
Since every induced subdivision of the diamond in $G$ intersects $H$, 
it intersects at least one of $P_1, P_2, P_3$.
So, by applying \autoref{prop:basedpath} to each of $P_1, P_2, P_3$, 
we can output in time $\mathcal{O}(N(3, 3k)^{3k}+k\abs{G}^7)$ 
either $k$ vertex-disjoint induced subdivisions of the diamond, or a vertex set of size at most $g_1(k)=3g_2(k)$ hitting every induced subdivision of the diamond.
\end{proof}

\begin{proof}[Proof of \autoref{t:diamond}]
We assign $g(k):=k g_1(k)$.
We can prove \autoref{t:diamond} using \autoref{prop:basedpath2}, with exactly the same argument in the proof for $1$-pan (\autoref{t:1panep}).
\end{proof}

\section{Concluding remarks and open problems}
\label{sec:concl}

In this paper, we investigated the induced \ep{} property of subdivisions beyond known results about cycles and obtained both positive and negative results.
We note that our positive results for pans come with polynomial-time algorithms that output either a large packing of induced subdivisions of the considered graph $H$, or a small hitting set. These can be directly used to design approximation algorithms for computing the maximum size of a packing of induced subdivisions of $H$ and the minimum size of a hitting set (as in \cite{Chatzidimitriou2017logopt} for instance). For 1-pans and 2-pans, this gives a polynomial-time $\mathcal{O}(\textsf{OPT}\log \textsf{OPT})$-approximation.
On the other hand, our negative results cover a vast class of graphs.

The most general open problem on the topic discussed in this paper is to characterize the graphs $H$ whose subdivisions have the induced \ep{} property.
According to \autoref{thm:badcases}, every graph $H$ for which the question is open satisfies the following:
\begin{itemize}
\item $H$ is planar and has a cycle $C$;
\item every induced cycle of $H$ is a $C_3$ or a $C_4$;
\item let $\bar{N}(C)$ denote the vertices of $H$ that are not
  adjacent to $C$ (equivalently, vertices at distance at least 2 from
  $C$), then $|\bar{N}(C)| \leq 2$ and in the case of equality, the
  two vertices of $\bar{N}$ are independent.
\end{itemize}
The study of subdivisions of specific graphs is an intermediate step towards this goal.
A direction of research towards the aforementioned characterisation
would be to investigate whether the \ep{} property is inherited by
induced subdivisions. Formally, is it true that if the induced \ep{} property holds
for subdivisions of some graph~$H$, then it also holds for the
subdivisions of every graph $H'$ contained as an induced subdivision
in~$H$?

Observe that the constructions we used in our counterexamples
contain arbitrarily large complete subgraphs. Therefore the landscape
of the induced \ep{} property of subdivisions might be much different
if one restricts their attention to graphs excluding a dense
subgraph. In this direction, Wei{\ss}auer recently proved
that for every $s,\ell\in \N$, subdivisions of $C_\ell$ have the
induced \ep{} property in $K_{s,s}$-subgraph-free
graphs~\cite{2018arXiv180302703W}. This contrasts with the
general case where subdivisions of $C_\ell$ stop having the
induced \ep{} property from $\ell = 5$ (see the note below \autoref{th:kim17}).

Another line of research in the study of the \ep{} property of graph classes is to optimize the bounding function.
We note that all our positive results hold with a polynomial bounding function. On the other hand, we obtained in \autoref{thm:omegaklogk} a lower bound of $\Omega(k\log k)$ for non-acyclic subcubic graphs. We do not expect our upper-bounds to be tight and it is an open question to find the correct order of magnitude of the bounding functions for the graphs we considered. In this direction it is also open to determine the correct order of magnitude of the bounding function in \autoref{th:kim17}.

\newcommand{\etalchar}[1]{$^{#1}$}

\end{document}